\documentclass[11pt]{article}
\usepackage{graphicx} 

\usepackage{amsmath}
\usepackage{amssymb}

\usepackage{pre}
\usepackage[colorinlistoftodos,prependcaption,textsize=tiny]{todonotes}

\title{Polylogarithmic-Weight Dicke States in QAC$^0$ and Arbitrary Symmetric States in QAC$^0_f$}

\def\ANON{0} 
\def\COMM{0} 

\ifnum\ANON=1
\author{
  Anonymous Authors.
}
\else
\author{
  Lucas Gretta \thanks{University of California at Berkeley. \ Email: \url{lucas_gretta@berkeley.edu}. \ Supported by NSF Award CCF-2231095}
  \and
  Meghal Gupta\thanks{University of California at Berkeley. \ Email: \url{meghal@berkeley.edu}. \ Supported by NSF GRFP.}
  \and
  Malvika Raj Joshi\thanks{University of California at Berkeley. \ Email: \url{malvika@berkeley.edu}. \ Supported by NSF grant 2311733 and DOE grant DE-SC0024124.} 
}
\fi
\date{}

\ifnum\COMM=1
\newcommand{\malvika}[1]{\textcolor{purple}{[\textbf{Malvika:} {#1}]}}
\newcommand{\luke}[1]{\textcolor{blue}{[\textbf{Luke:} {#1}]}}
\newcommand{\meghal}[1]{\textcolor{red}{[\textbf{Meghal:} {#1}]}}
\newcommand{\gpt}[1]{\textcolor{orange}{[\textbf{GPT:} {#1}]}}
\newcommand{\todomal}[1]{\todo[linecolor=Plum,backgroundcolor=Plum!25,bordercolor=Plum]{\textbf{@mal todo}: #1}}
\newcommand{\todoluke}[1]{\todo[linecolor=blue,backgroundcolor=blue!25,bordercolor=blue]{\textbf{@luke todo}:#1}}
\newcommand{\todomeghal}[1]{\todo[linecolor=red,backgroundcolor=red!25,bordercolor=red]{\textbf{@meghal todo} #1}}
\newcommand{\todoany}[1]{\todo[linecolor=orange,backgroundcolor=orange!25,bordercolor=orange]{\textbf{todo @plz do}: #1}}
\else
\newcommand{\malvika}[1]{}
\newcommand{\luke}[1]{}
\newcommand{\meghal}[1]{}
\newcommand{\gpt}[1]{}
\newcommand{\todomal}[1]{}
\newcommand{\todoluke}[1]{}
\newcommand{\todomeghal}[1]{}
\newcommand{\todoany}[1]{}
\fi

\begin{document}

\maketitle
\ifnum\ANON=1
\PackageWarningNoLine{Global}{Note anonymous mode}
\ifnum\COMM=1
\textcolor{red}{\textbf{Warning:}{ Comments are enabled in anonymous mode}}
\fi
\fi
\ifnum\COMM=1
\PackageWarningNoLine{Global}{Note comments are enabled}
\fi

\begin{abstract}

An $n$-qubit Dicke state of weight $k$, $\ket{D^n_k}$, is the uniform superposition over all $n$-bit strings of Hamming weight $k$. Dicke states are central to quantum algorithms exhibiting speedups, such as Decoded Quantum Interferometry (Jordan et al., \emph{Nature}, 2025). In the NISQ era, quantum hardware is constrained by both depth and locality, motivating the question of which global operations suffice to prepare such states. $\QACZ$, the quantum analogue of $\ACZ$, minimally extends local $O(1)$-depth quantum circuits by allowing arbitrary-width Toffoli (reversible AND) gates. We show that Dicke states of $\mathrm{polylog}(n)$ weight can be prepared in $\QACZ$. This gives the first $\QACZ$ construction of any super-constant-weight $n$-qubit Dicke state, since previous constructions relied on the much more powerful $\FANOUT_n$ gate. In general, we show that any weight-$k$ Dicke state can be constructed using $\FANOUT_{\min(k,n-k)}$ gates. Combined with recent hardness results, this yields a tight characterization: for $k \leq n/2$, $\ket{D^n_k}$ can be prepared in $\QACZ$ if and only if $\FANOUT_k \in \QACZ$.

We develop a limited-fanout state-synthesis toolkit for $\QACZ$ that yields further constant-depth, $\poly(n)$-ancilla constructions. Every $n$-qubit symmetric state supported on Hamming weight $\leq k$ can be prepared using $\FANOUT_k$ gates. Moreover, every $O(\log n)$-qubit state can be prepared using quantum random-access memory ($\QRAM_n$), which refers to a coherent indexing gate $\ket{x}\ket{i}\ket{0} \mapsto \ket{x}\ket{i}\ket{x_i}$, a potentially weaker resource than $\FANOUT_n$ ($\QRAM_n \in \QACZF$). Thus, every symmetric state can be synthesized exactly in $\QACZF$.

\end{abstract}
\tableofcontents

\section{Introduction}



$\QACZ$ circuits are a fundamental model of shallow quantum computation consisting of constant depth, polynomial size circuits made up of reversible-$\AND$ gates and single qubit unitaries. These circuits are the direct quantum analog of classical $\ACZ$ circuits \cite{moore1999qac0} but it is unknown if $\QACZ$ can simulate $\ACZ$. It is a major open question whether $\QACZ$ can implement the $\FANOUT_n$ operation, which, given a classical bit $b$, produces $n$ \emph{classical} copies of $b$. This operation is freely available in $\ACZ$, which cannot compute functions such as $\PARITY_n$ and $\MAJORITY_n$ \cite{hastad1986switch}. In contrast, the ability of $\QACZ$ circuits to perform classical tasks such as $\PARITY_n$, modular arithmetic, threshold computations, and sorting, depends on whether $\QACZ$ can implement $\FANOUT_n$ \cite{moore1999qac0, green2002acc, hoyer2005fanout}.

The hardness of $\FANOUT$ for $\QACZ$ is also closely tied to the complexity of several state-synthesis tasks, including the preparation of $n$-qubit GHZ states \cite{moore1999qac0, rosenthal2021qac0}. GHZ states are fundamental entangled resources with applications in quantum metrology and nonlocality tests \cite{giovannetti2011advances, monz201114, kafatos2013bell}. However, $n$-qubit GHZ states can be synthesized in $\QACZ$ if and only if $\FANOUT_n \in \QACZ$. This provides an inherently-quantum complete task for $\QACZF$, the class of constant depth circuits obtained by augmenting $\QACZ$ with fanout.  



Another important entanglement resource is the family of Dicke states, introduced by Robert H. Dicke \cite{dicke1954coherence}. The weight $k$, $n$-qubit Dicke state, $\ket{D^n_k}$, is defined as the uniform superposition over $n$-bit strings of Hamming weight $k$. These states form a complete orthonormal basis for the symmetric subspace of $n$ qubits, which comprises all states invariant under permutations of the qubits. Dicke states also have applications in quantum metrology \cite{zhang2014quantum} and play a central role in Decoded Quantum Interferometry (DQI) \cite{jordan2025optimization}. 
From a computational perspective, Dicke states capture the ability of a circuit to generate a superposition over inputs satisfying a high-degree combinatorial constraint.
This makes Dicke state preparation a natural candidate to probe the ability of shallow circuits to enforce global structure.

The states $\ket{D^n_{n-k}}$ and $\ket{D^n_k}$ are equivalent up to local operations; thus, we may assume $k \le n/2$ to discuss their circuit complexity in terms of $k$. For $k = O(1)$, these states can be prepared in $\QACZ$ \cite{cleandicke2026}. However, for any $k=\omega(1)$, existing constant-depth constructions rely on the $\FANOUT_n$ gate \cite{cleandicke2026} and cannot be implemented in $\QACZ$ unless $\QACZ = \QACZF$.
Recent work showed that $\FANOUT_k$ is in fact \emph{necessary} for preparing weight-$k$ Dicke states \cite{gretta2026paritynotinqac0iff}. However, prior to this work, it was not known whether $\FANOUT_k$ suffices to construct weight-$k$ Dicke states on $n$ qubits. In particular, for the case of $k=\polylog(n)$, although $\QACZ$ circuits can implement $\FANOUT_k$ \cite{rosenthal2021qac0, grier2024threshold}, it was not known how to construct $\polylog(n)$ weight Dicke states on $n$ qubits. Existing techniques for constructing constant weight Dicke states break down in this regime -- they either require Grover amplification in $\poly(k)$ depth or parallel amplification of $\poly(n)$ runs that crucially requires $\FANOUT_{\poly(n)}$ to aggregate the runs \cite{cleandicke2026, grier2026tc0, rosenthal2021qac0}.

In this work, we close this gap by proving new upper bounds for Dicke-state preparation: for every $k=\polylog(n)$, weight-$k$, $n$-qubit Dicke states can be constructed exactly and cleanly in $\QACZ$. More generally, we provide explicit $\QACZ$ circuits with access to $\FANOUT_k$ gates (instead of $\FANOUT_n$ gates) that prepare $\ket{D^n_k}$ and $\ket{D^n_{n-k}}$. 
In light of the reductions of \cite{gretta2026paritynotinqac0iff}, our 
upper bounds are \emph{optimal} with respect to the width of $\FANOUT$ gates. This establishes an equivalence between Dicke state preparation and $\FANOUT$ :  for any $k \leq n/2$,
$$\ket{D^n_k} \text{ can be prepared in } \QACZ \iff \FANOUT_k \in \QACZ .$$


In the process of showing the above result, we develop a limited-fanout state-synthesis toolkit for $\QACZ$ and obtain resource-efficient $O(1)$-depth constructions for more general state-synthesis. 

\paragraph{Arbitrary state synthesis and QRAM significance.} Quantum random-access memory (QRAM) provides coherent access to classical data \cite{hann2021resilience,qramarchitecture} and is a commonly used primitive in quantum machine learning \cite{Kerenidiskmeans,kerenidis_et_al2017}.
We use $\QRAM_n$ to denote a coherent indexing gate into a $n$-bit database, i.e., the map $\ket{x} \ket{i} \ket{0} \mapsto \ket{x} \ket{i} \ket{x_i}$ for $x \in \bin^n$ and $i \in [n]$. 
The corresponding classical function, $\mathsf{Indexing}(x,i)=x_i$, is computable by $O(n)$-size depth-$2$ $\ACZ$ circuits. 
Thus, $\QACZ[\QRAM_n]\subseteq \QACZ[\FAN_n]$, but it is unknown whether the converse also holds.

Our toolkit yields new constant-depth upper bounds for arbitrary state synthesis. Previously, arbitrary state synthesis in constant depth was only known through full quantum fanout, via the  $\poly(n)$-ancilla $\QACZ[\FAN_n]$ circuits for $O(\log n)$ qubit states \cite{rosenthal2021query}.  We show an analogous upper-bound using a potentially weaker resource: every $O(\log n)$-qubit state can be prepared in $O(1)$ depth by a $\poly(n)$-size $\QACZ[\QRAM_n]$ circuit. Building on this, we give the first $O(1)$-depth, $\poly(n)$-size unitary synthesis of \emph{arbitrary symmetric states}. Specifically, every $n$-qubit symmetric state supported on Hamming weights at most $k$ can be prepared by a $\poly(n)$-size $\QACZ[\FAN_k]$ circuit.

From a complexity perspective, $\QACZ[\QRAM] \subseteq (\QACZ)^{\ACZ} \subseteq\QACZF$. Our results show that $(\QACZ)^{\ACZ}$ can perform genuinely quantum tasks, beyond just classical reductions to $\ACZ$, and motivate further study of this intermediate model.




\subsection{Our results}

Our main result is that $\QACZ$ augmented with $\FANOUT_k$ gates, which we denote as $\QACZ[\FANOUT_k]$, can construct $\ket{D^n_k}$. Previously, $\ket{D^n_k}$ was only known to be preparable in constant depth with access to $\FANOUT_n$, rather than $\FANOUT_k$.


\begin{theorem}[$\ket{D^n_k}$ with $\FANOUT_k$ (Informal \cref{thm:kdicke})]\label{thm:kdicke_intro} 
For every $n$ and $k \le n/2$, the states $\ket{D^n_k}$ and $\ket{D^n_{n-k}}$ can be exactly and cleanly prepared by a constant depth $\poly(n)$-ancilla $\QACZ[\FANOUT_k]$ circuit. 
\end{theorem}

\noindent When $k = \polylog(n)$, $\FANOUT_k$ can be implemented using $\poly(n)$ ancillae in $\QACZ$ \cite{rosenthal2021qac0, grier2026tc0}. This immediately leads to the following corollary. 

\begin{restatable}[Poly-log weight Dicke states in $\QACZ$]{corollary}{polylogdicke}\label{thm:polylog_dicke_intro} 
For all $n$ and $k = \polylog(n)$, the states $\ket{D^n_k}$ and $\ket{D^n_{n-k}}$ can be synthesized exactly and cleanly in constant depth with $\poly(n)$-ancilla, using only multi-qubit Toffoli gates and single qubit unitaries.
\end{restatable}

We remark that the upper bound in \cref{thm:kdicke_intro} is \emph{tight} in the required $\FANOUT$, since previous work \cite{gretta2026paritynotinqac0iff} proved $\FANOUT_k$ is necessary to construct $\ket{D^n_k}$. Thus, we have the following equivalence. 

\begin{corollary}[Dicke states are equivalent to fanout]
   For any $n$ and $k\leq n/2$, the state $\ket{D^n_k}$ can be prepared in $\QACZ$ if and only if $\FANOUT_k \in \QACZ$.  
\end{corollary}

Our constructions require substantially different techniques than the constant-weight Dicke state construction of \cite{cleandicke2026}, as we elaborate upon in \Cref{sec:overview}.

Moreover, our results generalize to symmetric states that are supported on low weight Dicke states.

\begin{theorem}[General symmetric states with $\FANOUT_k$ (Informal \Cref{thm:any_symm})]\label{thm:symmetric_k_intro}
Let $k \le n$. Any $n$-qubit symmetric state $\ket{\psi}$ supported on weights at most $k$ can be constructed exactly and cleanly in constant depth with $\poly(n)$ ancillae by a $\QACZ[\FANOUT_k]$ circuit.
\end{theorem}

\noindent As a consequence, we obtain constant-depth $\poly(n)$-sized constructions for all symmetric states using $\FANOUT_n$.

\begin{corollary}[Arbitrary symmetric states in constant depth]\label{cor:symmetric_qaczf}
Every symmetric state on $n$ qubits can be prepared exactly and cleanly by a constant-depth $\poly(n)$ ancilla $\QACZF$ circuit.
\end{corollary}


Finally, we state the following result for constructing arbitrary states in constant depth with exponentially many ancillae using $\QRAM$, scaled down to $m=O(\log n)$ as follows. 

\begin{theorem}[General state synthesis with QRAM (Informal \cref{thm:anylognqram})]
Every $\ell = O(\log n)$-qubit state $\ket{\psi}$ can be cleanly prepared in $\QACZ[\QRAM_n]$ using $\poly(n)$ ancillae. 
\end{theorem}

\noindent Notably, \cite{rosenthal2021query} showed a comparable upper-bound, which, when scaled down, gives a $\poly(n)$-size $\QACZ[\FAN_n]$ circuit for every $O(\log n)$-qubit state. Our construction is meaningfully different and relaxes the resource requirement, since $\FAN_n$ gates are known to be at least as powerful as $\QRAM_n$ gates.
\subsection{Prior work}
$\QACZ$, the quantum analogue of $\ACZ$, is defined as the class of constant depth circuits with a polynomial number of ancilla qubits. The gate set of $\QACZ$ consists of $\poly(n)$-width Toffoli gates (reversible $\AND$) and single-qubit unitaries \cite{moore1999qac0}. Although wide Toffoli gates allow for global interactions among qubits, these interactions are seemingly weaker than those allowed by the $\FANOUT$ gate. For instance, $\FANOUT_n$ gates and local unitaries can simulate $\poly(n)$-width Toffoli gates in constant depth ($\QACZ \subseteq \QNC_f = \QACZF$) \cite{takahashi2012collapse} . However whether the converse, $\FANOUT_n \in \QACZ$, holds is a major open problem in quantum circuit complexity. 

Beyond specific state families, \cite{rosenthal2021query} showed that every $m$-qubit state can be prepared exactly by a constant-depth $\QACZF$ circuit (unbounded $\FANOUT$ gates) using $\widetilde{O}(2^m)$ ancillae. In particular, every $O(\log n)$-qubit state admits such a construction using $\FANOUT_n$ and $\poly(n)$ ancillae.


The work on Dicke states in the setting of $\QACZ$, which allows arbitrary-width gates, is fairly new. Meanwhile, there is extensive work on preparing Dicke states in the more restricted setting where the gates are subject to locality constraints that limit their width to $O(1)$. Upper bounds in this setting include \cite{bartschi2022detdicke,bartschi2022shortdicke,yuan2025depth,cruz2019w,liu2024low, raveh2024dicke}, all of which require $\Omega(\log n)$ depth as dictated by light-cone based lower-bounds \cite{moore2001parallel}. Achieving constant depth requires bypassing this lower bound and some prior works do so by allowing for (long-range) classical communication and post-processing. These result in non-unitary constructions that are incomparable to ours and include \cite{Buhrman2024statepreparation,piroli2024approximate,farrell2025digital, yu2025efficient}.

Light-cone lower-bounds break down in the $\QACZ$ setting in which Dicke states have recently been of interest for both algorithms \cite{grier2026tc0, cleandicke2026} and circuit complexity, specifically in understanding the power of Toffoli gates vs $\FANOUT$ ($\QACZ$ vs $\QACZF$) \cite{gretta2026paritynotinqac0iff}.
The first constant-depth unitary construction of a nontrivial Dicke state is due to \cite{grier2026tc0}, who gave a $\QACZ$ circuit for preparing $\ket{D^n_1}$, also known as the $W$ state. Following,  \cite{cleandicke2026} showed that for any $k = O(1)$, $\dkn$ can be prepared in $\QACZ$. However for $k = \omega(1)$, the only known constant depth unitary construction is that of \cite{cleandicke2026} and uses $\FANOUT_n$ to prepare $\ket{D^n_k}$, producing a $\QACZF$ circuit. As referenced earlier, $\QACZF$ is a very powerful class and is in fact the quantum analogue of $\TCZ$, i.e. $\QACZF = \QTCZ$ \cite{grier2026tc0}. It is unknown if $\QACZ$, the quantum analogue of $\ACZ$, can simulate $\QACZF$, or even $\ACZ$. This naturally raises the question of whether Dicke states can be prepared without relying on $\FANOUT_n$ gates. 
Recently, \cite{gretta2026paritynotinqac0iff} showed that $\FANOUT_k$ is necessary to prepare $\ket{D^n_k}$. However, for $k = \omega(1)$, known constructions rely on $\FANOUT_n$, leaving a large gap between the upper and lower bounds on the required fanout.

\section{Overview}\label{sec:overview}


In this section, we describe the main ideas behind the proof of our first main theorem and briefly touch on the extension to arbitrary symmetric states.

\begin{restatable*}[$\ket{D^n_k}$ with $\FANOUT_k$]{theorem}{kdicke}\label{thm:kdicke}
There exist constants $d \in \mathbb{N}$ and $c > 0$ such that for every $n \in \mathbb{N}$ and $k \in [n]$, there is a clean depth-$d$ $\QACZ[\FANOUT_k]$ circuit of size at most $n^c$ that takes as input $\ket{0^n}$ and outputs $\ket{D^n_k}$.
\end{restatable*}

We prove this theorem in \cref{sec:dicke} and provide an overview below. In \cref{sec:primitives}, we introduce a toolkit of primitives for constant-depth state synthesis, with particular emphasis on limited fanout. The first set of primitives we introduce enables controlled operations for certain structured circuits. Note that implementing a controlled version of an otherwise arbitrary $\QACZ$ circuit cannot necessarily be done in $\QACZ$ (See \cref{sec:ctrlunitaries}). We also introduce additional primitives to manipulate amplitudes by rescaling branches (\cref{lem:adjustamp}). Together, these tools allow us to construct more complex states. In particular, we show that every $O(\log n)$-qubit state can be prepared cleanly in $\QACZ[\QRAM_n]$ using $\poly(n)$ ancillae (\cref{thm:anylognqram}).





\paragraph{Existing tools and their limitations.} We start by explaining the approach used in \cite{cleandicke2026} to construct constant-weight Dicke states. 

\cite{grier2026tc0} proved that Grover's amplitude amplification procedure \cite{Grover} can be performed in $\QACZ$. This can be used to amplify any state of the form $\sqrt{\alpha} \ket{\psi^*}\ket{1} + \sqrt{1-\alpha} \ket{\bad}\ket{0}$ with at least a constant $\alpha$ weight on the marked part, to the marked state $\ket{\psi^*}$ in constant depth. \cite{grier2026tc0} also showed that $\THRESHOLD^n_k$, and therefore $\EXACT^n_k$, can be implemented using $\FANOUT_k$ in $\QACZ$. 

One approach to constructing a Dicke state is to start with the state $\ket{k/n}^{\tens n}\ket{0}$ ($\ket{p}:=\sqrt{1-p}\ket{0}+\sqrt{p}\ket{1}$.) Using the $\EXACT^n_k$ gate, we mark the component with exactly $k$ 1's to produce a state of the form $\sqrt{\alpha} \dkn\ket{1} + \sqrt{1-\alpha} \ket{\bad}\ket{0}$. From here, if $\alpha$ is constant, we can apply Grover's amplitude amplification procedure a constant number of times to produce $\dkn$ exactly. Since $\alpha$ corresponds exactly to the probability of getting exactly $k$ ones in $n$ independent $\text{Bernoulli}(k/n)$ samples, $\alpha$ is indeed constant when $k$ is constant. 

However, this idea falls short as soon as $k = \omega(1)$. Indeed, there is no similar $O(1)$-depth amplification procedure anymore, and that is the central challenge that we must overcome.


\subsection{Our construction}

In order to use the $\EXACT^n_k$ gate and Grover's amplitude amplification, we need to construct some state of the form $\sqrt{\alpha} \dkn \ket{1} + \sqrt{1-\alpha} \ket{\bad}\ket{0}$ where $\alpha$ is constant and $\ket{\bad}$ only has amplitude on strings of weight not equal to $k$.

\noindent It suffices to construct some $\ket{\psi}$ that meets the following two criteria:

\begin{enumerate}[label=(\arabic*), leftmargin=3em]
    \item \label{itm:const-fid} $\ket{\psi}$ has constant fidelity with $\dkn$.
    \item $\ket{\psi}$ has equal amplitude on all strings with $k$ ones.
\end{enumerate}

Note the necessity of the second constraint. If $\ket{\psi}$ had different amplitudes on the different weight-$k$ strings, then applying the $\EXACT^n_k$ gate to $\ket{\psi}$ would not result in a marked branch exactly equal to $\dkn$, which is needed for the amplitude amplification step. In particular, it is worth noting that \ref{itm:const-fid} may be satisfied even if some weight-$k$ strings are entirely missing from the support of $\ket{\psi}$.

\subsubsection{Constructing a state with constant fidelity with \safedkn.}
For simplicity, assume that $n \geq k^3$ and that $k^3$ divides $n$; the general parameter regime is handled in the formal proof.
We will start with the simpler goal of just constructing a state $\ket{\psi}$ that has constant fidelity with $\dkn$, i.e. a state that satisfies \ref{itm:const-fid}.
Although known constructions of $\dkn$ require $\FANOUT_n$, we can apply them to prepare $\ket{D^{k^3}_k}$ using $\FANOUT_k$ \cite{cleandicke2026}.
Starting with this state, we will ``expand'' each qubit into $m=n/k^3$ qubits, so that ultimately the state will be on $n$ qubits.

More formally, we will make use of the following controlled-W map, which is known to be in $\QACZ$, for example as a consequence of \cref{lem:ctrlstate} and \cref{fact:zerow}. 
\begin{itemize}
    \item $\ket{0}\ket{0^m} \to \ket{0}\ket{0^m}$
    \item $\ket{1}\ket{0^m} \to \ket{0}\ket{W_m}$
\end{itemize}

First, partition the $n$ target qubits into $\ell = k^3$ buckets of size $m$. On $\ell$ ancilla qubits, construct $\ket{D^{\ell}_k}$, where each qubit corresponds to one of the $\ell$ buckets. Then, for each bucket, whenever the corresponding ancilla qubit is a $1$, map it to $\ket{W_m}$ on the target qubits within that bucket, and otherwise to $\ket{0^m}$. This procedure is illustrated below in \Cref{fig:bucket_pic}.

\usetikzlibrary{matrix,decorations.pathreplacing,positioning,calc}
\newcommand{\mycircle}{\tikz[baseline=-0.6ex]\draw[line width=0.08ex] (0,0) circle (0.9ex);}
\newcommand{\myfilledcircle}{\tikz[baseline=-0.6ex]\fill (0,0) circle (0.9ex);}
\begin{figure}[ht]
\centering
\begin{tikzpicture}[every node/.style={anchor=base}]
\matrix (m) [matrix of math nodes,row sep=1em,column sep=0.8em,row 4/.style={nodes={yshift=-1.8em}}]{
\myfilledcircle & \longrightarrow & \mycircle & \myfilledcircle & \mycircle & \dots & \mycircle \\
\mycircle & \longrightarrow & \mycircle & \mycircle & \mycircle & \dots & \mycircle \\
\vdots & & & & & & \\
\myfilledcircle & \longrightarrow & \mycircle & \mycircle & \myfilledcircle & \dots & \mycircle \\
};
\draw[decorate,decoration={brace,amplitude=4pt,mirror}]
($(m-1-1.north west)+(-0.1em,0.1em)$) -- ($(m-4-1.south west)+(-0.1em,-0.1em)$)
node[midway,left=6pt] {$\ket{D^{k^3}_k}$};
\draw[decorate,decoration={brace,amplitude=4pt}] ($(m-1-3.north west)+(0,0.1em)$) -- ($(m-1-7.north east)+(0,0.1em)$) node[midway,above=2pt] {$\ket{W_m}$};
\draw[decorate,decoration={brace,mirror,amplitude=4pt}] ($(m-2-3.south west)+(0,-0.1em)$) -- ($(m-2-7.south east)+(0,-0.1em)$) node[midway,below=2pt] {$\ket{0^m}$};
\draw[decorate,decoration={brace,amplitude=4pt}] ($(m-4-3.north west)+(0,0.1em)$) -- ($(m-4-7.north east)+(0,0.1em)$) node[midway,above=2pt] {$\ket{W_m}$};
\end{tikzpicture}
\caption{Visualization of the bucket occupancy of $\ell=k^3$ buckets (left) and the corresponding $n=mk^3$-qubit state on each bucket obtained via a controlled-W preparation on each group of $m$ qubits (right).}\label{fig:bucket_pic}
\end{figure}
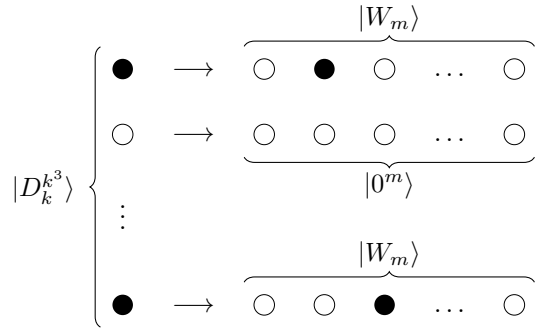

This gives us an equal superposition over all strings of weight $k$ \emph{where each bucket has at most one 1}. In order to show that this state has constant fidelity with $\dkn$, it suffices to show that a uniformly random weight-$k$ string has constant probability of having at most one $1$ in every bucket. This is true by a simple birthday paradox argument: as long as $\ell=\omega(k^2)$, it is unlikely that two 1s will be in the same bucket. In fact, since a collision occurs with $o(1)$ probability, only an $o(1)$-fraction of weight-$k$ strings are missing from our superposition.

At a high level, this method overcomes the challenge of making strings of weight exactly $k$ appear with constant probability by using the $\ket{D^{\ell}_k}$ construction. However, in the process, we lose the small fraction of strings that don't adhere to the induced bucket structure, which violates our second criterion.

\subsubsection{Adding back the missing strings.}

If there are $j<k$ non-empty buckets in a given string, it will not appear in our superposition. Let $p(j)$ denote the probability that a uniformly random weight-$k$ string occupies exactly $j$ of the $\ell$ buckets. More precisely, we can write $\dkn$ as follows:

\begin{align}
\dkn &= \sum_{j = 1}^k \sqrt{p(j)} \dmkj
\end{align}

Here, $\dmkj$ is the superposition of all weight $k$ strings where the ones fall into $j$ of the $\ell$ buckets. The previous section approximated the full state by constructing only the $j=k$ case $\ket{D^{m,\ell}_{k,k}}$. 

In the next section, we will describe a more general construction that also works for each $j<k$. We remark that this case was especially nice because all non-empty buckets have exactly one 1, but this property does not hold in general. Moreover, the construction we described for the $j=k$ case was exact, i.e. only had strings of weight $k$, but the general argument described below will produce strings with a $\ket{\bad}$ component that we will need to remove with amplitude amplification.

\paragraph{Constructing strings for a specific $j<k$.}

For a given $j<k$, we similarly start by constructing $\ket{D^{\ell}_j}$ on ancilla qubits, where the $1$'s correspond to the nonempty buckets. We need to modify our construction so that a non-empty bucket can have more than one 1. Instead of using a map that takes $\ket{1} \to \ket{1}\ket{W_m}$, we will use a map which allows each bucket to have more than one 1 with some small probability. More specifically, define $\ket{\varphi(s(1),\ldots,s(k))} := \sum_i \sqrt{s(i)}\ket{D^m_i}$. As an example, the W state has $s(1)=1$ and all other $s(i)=0$. Then, we construct the following map for some specific probabilities $s(1),\ldots,s(k)$ summing to 1 that we'll solve for later: 
\begin{itemize}
    \item $\ket{0}\ket{0^m} \to \ket{0}\ket{0^m}$
    \item $\ket{1}\ket{0^m} \to \ket{0}\ket{\varphi(s(1),\ldots,s(k))}$
\end{itemize}
We do not know how to construct this state for arbitrary values of $s(1),\ldots,s(k)$ (in fact, this is already more general than \Cref{thm:kdicke_intro} for $m$ instead of $n$), so we will have to pick them specifically such that we can construct this state. Intuitively, the state we will require can be constructed because the distribution we will end up needing is very close to that of the W state where $s(1)=1$ and all other $s(i)=0$.

Unlike the $j=k$ case, in this construction, it will not always hold that the resulting state has exactly $k$ ones, and we will need to post-select on such strings. We can mark all such strings, but we can only use Grover's amplitude amplification to post-select if the marked branch has constant amplitude. We might hope that we can achieve this by setting  $s(1),\ldots,s(k)$ appropriately, but actually, as it turns out, when $j\ll k$, there is actually \emph{no choice} of  $s(1),\ldots,s(k)$ that works, but we will be able to get around this by post-selecting only at the very end, and not for each choice of $j$, since that way we only require that the combined $\ket{\bad}$ parts are constant weight.

\subsection{Parameter choices and putting it together.}

Let us outline the steps of our construction more precisely:
\begin{enumerate}[label=(\arabic*), leftmargin=3em]
    \item For some to-be-determined probability distribution $t(1),\ldots,t(k)$, construct a superposition of weight-$j$ Dicke states on $\ell=k^3$ ancilla qubits $\sum_j\sqrt{t(j)}\ket{D^{\ell}_j}$. (Think of $t(j)$ as the mass assigned to having $j$ nonempty buckets; this is not the same as $q(j)$ for reasons described later.) 
    \item For each such ancilla qubit corresponding to a bucket, map it to $m$ target qubits via the map $\ket{0}\ket{0^m} \to \ket{0}\ket{0^m}$ and $\ket{1}\ket{0^m} \to \ket{0}\ket{\varphi(s(1),\ldots,s(k))}$. (Although the $s(i)$'s could in theory depend on $j$, we won't need this freedom. For the special case of $j=k$, we had previously set $s(1)=1$ and $s(i)=0$ for $i>1$ in the example, but our final construction will set parameters differently.)
    \item Post-select on strings of weight $k$ using Grover's amplitude amplification.
\end{enumerate}

\noindent It remains to choose $t(1),\ldots,t(k)$ and $s(1),\ldots,s(k)$ to satisfy a few conditions:
\begin{itemize}
    \item We can implement the aforementioned map $\ket{0}\ket{0^m} \to \ket{0}\ket{0^m}$ and $\ket{1}\ket{0^m} \to \ket{0}\ket{\varphi(s(1),\ldots,s(k))}$.
    \item The probability of each string with $k$ ones is exactly the same.
    \item The total mass on strings of weight $k$ is constant (so that we can use Grover's amplitude amplification).
\end{itemize}

We will set $s(i):= C \cdot k\binom{m}{i}/(mk)^i$, where $C$ is chosen such that the $s(i)$ add to 1 (see \Cref{def:damped_binomial}). (It turns out that $C\approx 1$ always.) We will hold off on specifying the $t(j)$'s.

\paragraph{Implementing the map.}
We start by constructing the map $\ket{0}\ket{0^m} \to \ket{0}\ket{0^m}$ and $\ket{1}\ket{0^m} \to \ket{0}\ket{1/m}^m$. This gives us some ``approximate'' version of the state we want. Then, we mark each string with some qubit of the form $\sqrt{1-\alpha}\ket{0}+\sqrt{\alpha}\ket{1}$ to ``correct'' the ratios of amplitudes. For example, we mark the all $0$ branch with $\ket{0}$ since that isn't in the support of $\ket{\varphi(s(1),\ldots,s(k))}$, and similarly for all $\ket{D^{m}_{>k}}$ branches. We mark the $\ket{W_m}=\ket{D^{m}_1}$ branch with $\ket{1}$. For every other $\ket{D^{m}_{i}}$ $(2 \leq i \leq k)$, we mark it with some $\sqrt{1-\alpha}\ket{0}+\sqrt{\alpha}\ket{1}$ so that within the marked branch, the ratio of the amplitudes on the $\ket{D^{m}_i}$ and $\ket{D^{m}_1}$ branches is $\sqrt{s(i)/s(1)}$. We show in \cref{cl:hybridtoocc} that a constant fraction of the state will be marked, and so we can use Grover's amplitude amplification to construct $\ket{\varphi(s(1),\ldots,s(k))}$ exactly. (And we can ensure that the $\ket{0}\ket{0^m} \to \ket{0}\ket{0^m}$ part of the map is unaffected.)

\paragraph{Each weight-$k$ string has the same probability.}
Each weight-$k$ string $x$ has some particular number of non-empty buckets $j$, when its indices are partitioned into groups of $m$ as we have been doing.

Checking this condition really has two parts: $x$ and $x'$ with the same $j$ have the same probability, and \emph{after the post-selection step}, the probability of sampling a string with $j$ nonempty buckets is exactly $p(j)$. Essentially, the first condition ensures that the constructed state is some superposition over $\dmkj$'s while the second condition ensures it is the correct superposition. The former property follows directly from the $s(i)$'s being of the form $\lam \binom{m}{i}\beta^i$ (\cref{cl:shiftedoccuniform}). 

The latter property requires setting the $t(j)$'s appropriately. Naively, it would seem that we should just set $t(j):=p(j)$. However, each value of $j$ may have a different probability of post-selecting a weight-$k$ string, so the post-selection ratios may be different from the pre-selection ones. Define $q(j)$ to be the probability of post-selecting a weight-$k$ string when having $j$ non-empty buckets (i.e., the mass of the constructed state on $\dkmj$ when starting with $\ket{D^\ell_j}$ on the ancilla). Then, we set $t(j)\propto p(j)/q(j)$, and indeed the post-selection probabilities will be proportional to the $p(j)$. 

\paragraph{Total weight $k$ mass is constant.}
Finally, we need to check that the total mass on weight $k$ strings is constant, so that the final amplitude amplification step is possible. We discussed in the previous section that we can't just do amplitude amplification on each $j$ separately, because the $q(j)$'s will not always be constant for any choice of the $s(i)$'s and certainly not for the ones we chose. However, when we do amplitude amplification at the end, we only need to ensure that the total mass over all the $j$'s of the weight $k$ branch is constant. Each $j$ contributes $t(j)q(j)$ to this total mass, because its starting mass is $t(j)$ and the conditional probability of the final string being weight $k$ is $q(j)$. Thus, we just need to ensure that $\sum_j t(j)q(j) = \Theta(1)$, which is done in \Cref{cl:hybridtoocc}.

Thus, each step of the outlined algorithm indeed works and we can construct $\ket{D^n_k}$ exactly.

\subsection{Extending to arbitrary symmetric states}
We will briefly describe the ideas involved in the proof of \cref{thm:symmetric_k_intro}. The formal version is stated below.

\begin{restatable*}[Superposition of Dicke states inside Hamming ball $k$ with $\FANOUT_k$]{theorem}{ksymmstate}\label{thm:any_symm}
Let $\ket{\psi}$ be an arbitrary $n$-qubit symmetric state supported only on strings of Hamming weight at most $k$, of the form,
$$\ket{\psi} = \sum_{j = 0}^k \eta_j \ket{D^n_j}.$$ 
Then, there is a constant depth, $\poly(n)$-ancilla $\qaczfkd$ circuit to exactly and cleanly prepare $\ket{\psi}$. 
\end{restatable*}

Observe that symmetric states have compact $O(\log n)$-qubit descriptions that encode the amplitudes of the corresponding Dicke state superpositions. For example, the amplitudes of an arbitrary symmetric state $\ket{\psi}$ as above can be compacted into a $\log n$ qubit state $\ket{\wh{\psi}} := \sum_{j = 0}^k \eta_j \ket{j}$. In \cref{thm:anylognqram}, we give $\poly(k)$-ancilla preparations of arbitrary $O(\log k)$-qubit states using $\QRAM_k$. Since $\FANOUT_k$ can simulate $\QRAM_k$, this also gives $\poly(k)$-ancilla $\QACZ[\FAN_k]$ circuits for arbitrary $O(\log k)$ qubit states, which we use here. Alternatively, we may also use a scaled-down version of \cite{rosenthal2021query}'s construction of $\exp(n)$-size $\QACZF$ circuits for arbitrary $n$-qubit states.

Our Dicke state construction from \cref{thm:kdicke} also extends to a controlled circuit that prepares $\ket{D^n_k}$ on a fresh ancilla iff the control bit is $1$. Given the corresponding $O(\log k)$-qubit state, $\ket{\wh{\psi}}$, it is possible to construct $\ket{\psi}$ using the controlled-preparation of each $\ket{D^n_j}$ as a black-box. However, this is not sufficient for the fanout-efficient construction we claim in  \cref{thm:any_symm}, which uses $\FANOUT_k$ rather than $\FANOUT_n$. We instead prove \cref{thm:any_symm} by constructing the desired state from scratch using the bucketing strategy previously described. This proof follows much of the same strategy as our construction for $\dkn$ and we provide it in \cref{sec:arbitrarysymm}.

\section{Preliminaries}
The gateset of $\QAC^0$ is canonically defined as multi-qubit Toffoli (Reversible-AND) gates and single-qubit unitaries \cite{moore1999qac0} and the depth is given by the number of multi-qubit layers. Equivalent formulations include any multi-qubit reflection gate of the form $I - 2\kb{\vth}$ for arbitrary product state $\ket{\vth}$ \cite{rosenthal2021qac0}. Each qubit belongs to at most one gate per layer and just as classical reversible circuits, $\QACZ$ does not by default include $\FANOUT$, the ability to make (classical) copies of a qubit. A $\QACZ$ circuit is defined to have constant depth and $\poly(n)$ ancillae, and thus $\poly(n)$ gates. 
 
\begin{definition}[Circuit implementing unitary]
A circuit $C$ \emph{(cleanly) implements} an $n$-qubit unitary $U$ using $m = m(n)$ ancillae, if for all $n$-qubit states $\ket{\varphi}$, 
$$C \ket{\varphi} \ket{0^m} = \lr{U \ket{\varphi}} \ket{0^m}.$$
\end{definition}
\noindent More generally, for linear subspaces $\mathcal{S}, \mathcal{T}$ and a linear map $M : \mathcal{S} \to \mathcal{T}$, a circuit $C$ implements $M$ if it cleanly implements any unitary $U$ that behaves as described on $\mathcal{S}$.

\begin{definition}[State preparation circuit]
A circuit $C$ \emph{(cleanly) prepares} an $n$-qubit state $\ket{\psi}$ using $m = m(n)$ ancillae, if, $C \ket{0^n} \ket{0^m} =  \ket{\psi} \ket{0^m}.$
\end{definition}
\noindent All the upper bounds described in this work are \emph{clean computations}, i.e, the ancillae are returned to $\ket{0^m}$. Thus, any state preparation/unitary implementation circuits we refer to are presumed to be clean unless otherwise specified.


\begin{definition}[$\FANOUT_k$]\label{def:fanout}
   $\FANOUT_k$ is the reversible map satisfying $\ket{0} \ket{0^k} \mapsto \ket{0} \ket{0^k}$ and $\ket{1} \ket{0^k} \mapsto \ket{1} \ket{1^k}$. 
\end{definition}
Some works define $\FANOUT_k$ as a unitary on the entire space by the map $\ket{b} \ket{x} \mapsto \ket{b} \ket{x \oplus b^k}$ for $b \in \bin$ and $x \in \bin^n$, which turns out to be the Hadmard conjugation of $\PARITY_n$. 
The two definitions are equivalent in constant depth under $\QACZ$ reductions and for our purposes, \cref{def:fanout} suffices. Note that definition of $\qaczfkd$ remains unchanged regardless of which convention is used.

\begin{definition}[$\qaczfkd$ circuit]\label{def:qacwithfanout}
For a function $k = k(n)$, a circuit family $\clr{C_n}_n$ is a $\qaczfkd$ circuit family if for constants $c,a,n_0$, each $C_n$ for $n \geq n_0$ can be implemented using at most $n^{a}$ ancillae and $c$ layers of Toffoli gates, single qubit unitaries and $\FANOUT_{k(n)}$ gates.  
\end{definition}

\begin{definition}[$\QRAM_n$]\label{def:qramn}
The $\QRAM_n$ gate is the reversible map satisfying, for every $i\in[n]$, $x\in\bin^n$, and $b\in\bin$,
$$\ket{i}\ket{x}\ket{b} \mapsto \ket{i}\ket{x}\ket{b\oplus x_i}.$$
\end{definition}
\noindent  By recursively partitioning the memory into blocks of size $k$, $\QRAM_{\poly(k)}$ can be implemented using $\QRAM_k$ gates in $O(1)$ depth and $\poly(n)$ ancillae. Thus, analogous to $\FANOUT$, 
\begin{align}
\QACZ[\QRAM_k] = \QACZ[\QRAM_{\poly(k)}]. 
\end{align}

\begin{definition}[$\qaczfkd$ circuit]\label{def:qacwithqram}
For a function $k = k(n)$, a circuit family $\clr{C_n}_n$ is a $\QACZ[\QRAM_k]$ circuit family if for constants $c,a,n_0$, each $C_n$ for $n \geq n_0$ can be implemented using at most $n^{a}$ ancillae and $c$ layers of Toffoli gates, single qubit unitaries and $\QRAM{k(n)}$ gates.  
\end{definition}
\noindent This $\QRAM$ may be equivalently described as a restricted map with a clean target register, $b = 0$, since it can simulate \cref{def:qramn} by computing $x_i$ into a fresh ancilla, copying it to the target, and uncomputing the ancilla.

\paragraph{Circuit depth and ancilla notation.} We focus on state preparation problems for which the circuits $C_n$ are parameterized by the number of qubits $n$ in the \emph{output}. All $\QACZ[\cdots]$ circuits refer to $O(1)$ depth circuits, and the size/number of ancillae are implicitly taken to be $\poly(n)$ unless otherwise specified.  

\subsection{Additional Notation}\label{sec:gatenotation}
All the quantum states referred to are taken to be normalized unless otherwise specified. We use the following conventions.
\begin{enumerate}
    \item For $k \in \mathbb{N}$, $[k] := \clr{i}_{i = 1}^k$. For $i \in [k]$, $\ket{i}$ denotes the $\lceil \log k \rceil$-qubit standard basis vector, the binary representation of $i$ and $e^k_i$ refer to the string $0^{i-1} 1 0^{k-i-1}$, the one-hot encoding of $i$. When $k$, the number of bits, is clear from the context we denote it by $e_i$. Additionally, define $\delta_{ij} := \braket{e_i | e_j}$. 
    \item For $0 \leq \eps \leq 1$, we define $\ket{\eps}$ to be the single-qubit state $\ket{\eps} = \sqrt{1-\eps} \ket{0} + \sqrt{\eps} \ket{1}$. This is consistent with the usual definitions of $\ket{0}, \ket{1}$.
    \item  For $0 \leq \gamma \leq 1$, define $\Rot_{\gamma}$ to be the single-qubit unitary, 
    $\Rot_\gamma = \begin{bmatrix}
        \sqrt{\gamma} & \sqrt{1-\gamma} \\ 
        \sqrt{1-\gamma} & -\sqrt{\gamma}
    \end{bmatrix}$.
Note that any Hermitian unitary, such as $\Rot_{\gamma}$ is a reflection.
\item We will use the notation $\R_{\ket{\psi}}$ to denote the reflection about $\ket{\psi}$ given by $\R_{\ket{\psi}} := I - 2 \kb{\psi}$.
When $\ket{\psi}$ is separable this is a valid $\QACZ$ gate. Thus, for any fixed $\y \in \bin^n$, or any $i \in [n]$, both $\R_{\ket{\y}}$  and $\R_{\ket{i}}$ are valid $\QACZ$ gates.  
\item The $n$-qubit $\AND$ gate (Toffoli) is given by $(I - 2\kb{1^n} \tens \kb{+}_t)$ where $t$ is the output register. The $\OR$ gate is given by $(I - 2\kb{0^n} \tens \kb{+}_t)$, and the $\NOR$ gate by $X_t \cdot  (I - 2\kb{0^n} \tens \kb{+}_t)$
\item For a unitary $U$, $\ctrl{U}$ refers to the controlled-$U$ unitary given by, 
\begin{align}
    \ctrl{U}(x,T) = \kb{0}_x \tens I_T + \kb{1}_x \tens U_T
\end{align}
When $U$ is a valid $\QACZ$ multi-qubit gate, or a \emph{Hermitian} single qubit unitary,  $\ctrl{U}$ is a also a valid $\QACZ$ gate. 
Note that however, we cannot necessarily implement $\ctrl{C}$ for an arbitrary $\QACZ$ \emph{circuit} $C$ in $\QACZ$, as this requires $\FANOUT$. Although, given a $\FANOUT_k$ gate, we can implement $\ctrl{\FANOUT_k}$ by introducing an ancilla that computes the $\AND$ of the control qubit and the fanout source bit and then applying $\FANOUT_k$ on this ancilla.
\item For any binary string $\x \in \bin^n$, we use $\vlr{\x}$ to denote its Hamming weight.  
\item We use $[.]$ to denote the truth value of the expression inside the brackets.  
\end{enumerate}
Using this notation we define $\THRESHOLD^n_k$ and $\EX^n_k$. 
\begin{definition}[$\THRESHOLD^n_k$]\label{def:thresh}
    $\THRESHOLD^n_k$ is the reversible map satisfying for all $x \in \bin^n$, $$\ket{x} \ket{0} \mapsto \ket{x} \ket{\blr{|x| \geq k}}.$$ 
\end{definition}
\begin{definition}[$\EX^n_k$]\label{def:exactk}
    $\EX^n_k$ is the reversible map satisfying for all $x \in \bin^n$, $$\ket{x} \ket{0} \mapsto \ket{x} \ket{\blr{|x| = k}}.$$ 
\end{definition}
\noindent Observe that $\THRESHOLD^n_{n-k}$ can be implemented from $\THRESHOLD^n_k$ and also $\EX^n_k$ can be implemented using $\THRESHOLD^n_k$ and $\THRESHOLD^n_{k-1}$. 

\subsection{Known procedures in \safeqacz}
We list some standard procedures in $\QACZ$ and some simple consequences that we will utilize.

\begin{fact}[Reflection about preparable states]\label{fact:reflstate}
Let $\ket{\psi}$ be an $n$-qubit state, and $U$ be a unitary that cleanly prepares $\ket{\psi}$ using $m$ ancillae.  $\R_{\ket{\psi}}$ can be cleanly implemented with a depth $3$, $m$ ancillae circuit consisting of the gates $U, U^\dag$ and $I - 2\kb{0^{m+n}}$. 
\end{fact}
\begin{proof}
We will argue that $\R_{\ket{\psi}}$ is implemented using $m$ ancillae by the circuit, 
\begin{align}
C &:= U (I - 2\kb{0^{m+n}}) U^\dag \\
  &=  UU^\dag  - 2 U \kb{0^{m+n}} U^\dag \\
  &=  I - 2 \kb{\psi} \tens \kb{0^m}. 
\end{align}
Observe that the commutator $[I \tens \kb{0^m}, C] = 0$. Thus, for any $n$-qubit state $\ket{\varphi}$, as claimed,
\begin{align}
    C \ket{\varphi}_X \ket{0^m}_A &= \kb{0^m}_A \cdot C \cdot \ket{\varphi}_X \ket{0^m}_A \\
    &=  \kb{0^m}_A \cdot (I - 2 \kb{\psi}_X \tens \kb{0^m}_A) \cdot \ket{\varphi}_X \ket{0^m}_A \\
    &= \lr{(I - 2\kb{\psi}) \cdot \ket{\varphi}}_X \tens \ket{0^m}_A  \\
    &= \lr{\R_{\psi} \ket{\varphi}}_X \tens \ket{0^m}_A.
\end{align}
\end{proof}
\noindent The state $\ket{D^n_1}$ is also known as the $W$-state and can be prepared in $\QACZ$ \cite{grier2026tc0} as well the following states. 
\begin{fact}[Uniform $0^n$ and $W_{n}$ (Claim 5.1 of \cite{gretta2026paritynotinqac0iff})]\label{fact:zerow}
For any $n$, the state $\frac{1}{\sqrt{2}} \ket{0^n} + \frac{1}{\sqrt{2}} \ket{D^n_1}$ can be cleanly prepared in $\QACZ$ using $O(n)$ qubits. 
\end{fact}

\begin{fact}[Constant weight Dicke states in $\QACZ$ \cite{cleandicke2026}]\label{fact:constdicke}
For any $n$ and $k = O(1)$, the states $\ket{D^n_k}$ and $\ket{D^{n}_{n-k}}$ can be cleanly prepared in $\QACZ$.
\end{fact}

\noindent $\QACZ$ is not defined to include any $\FANOUT$ or $\THRESHOLD$ gates by default, but we can implement certain regimes in $\QACZ$.
\begin{fact}[Fanout and Thresholds in $\QACZF$ \cite{rosenthal2021qac0, grier2026tc0}]\label{fact:kthresh}
For any $m = \poly(k)$, $\FANOUT_m$ can be implemented using $O(1)$ $\FANOUT_k$ gates and, for any $n \geq m$, 
\[\THRESHOLD^{n}_m \in \qaczfkd.\]
Moreover, for any $n$ there is a $\poly(n)$-ancillae $\QACZ$ circuit for $\FANOUT_{\polylog(n)}$, i.e, \[ \qaczfk{\polylog(n)} = \QACZ .\]
\end{fact}

\subsubsection{Applications of \safefanout}

\begin{corollary}[Parallel W-controlled-swap Fanout (from Cor 3 of \cite{cleandicke2026})]\label{fact:ctrlswap}
Let $Q_1 \dots Q_m, Q_*$ be a set of $m+1 = \poly(k)$ $k$-qubit registers and $A$ be an $m$-qubit register. Then, the following map defined for every $i \in [m]$ and any state $\ket{\psi}$ on the $Q_1 \dots Q_m, Q_*$ registers, can be implemented in $\qaczfkd$.
$$\ket{e_i}_{A} \ket{\psi} \mapsto \ket{e_i}_A \swap(Q_i, Q_*) \ket{\psi},$$
where $e_i = 0^{i-1} 1 0^{m-i}$. In other words, this map swaps the registers $Q_i$ and $Q_*$ controlled on the $i$th qubit bit of $A$, under the promise that the state on $A$ has hamming weight $1$. 
\end{corollary}
\begin{proof}
Corollary 3 of \cite{cleandicke2026} states that the above map can be implemented in $O(1)$ depth whenever $Q_i$ is a single-qubit register. Then, using $\FANOUT_k$ we can make $m$ copies of each control $a_i$ to implement the map in parallel for each qubit of $T$.  
\end{proof}
\noindent Although not explicitly stated in their result, \cite{cleandicke2026} only use \emph{Hermitian} single qubit unitaries for all their constructions. We require this subtle point as stated below.
\begin{fact}[$\ket{D^n_k}$ with $\FANOUT_n$ \cite{cleandicke2026}]\label{fact:fanoutdicke}
There exists a $\QACZF=\qaczfk{n}$ circuit using $\poly(n)$ Hermitian gates that cleanly prepares the $n$-qubit state $\ket{D^n_k}$ for any $k \leq n$. 
\end{fact}
\noindent Note that \cref{fact:kthresh} and \cref{fact:fanoutdicke} together for $k = \polylog(n)$ only implies $\QACZ$ circuits for Dicke states on $\polylog(n)$ qubits, $\ket{D^{\poly(k)}_k}$, rather than $\ket{D^n_k}$.

\subsubsection{Amplitude amplification}
\begin{fact}[Exact Grover in $\QACZ$ (Thm 7 of \cite{grier2026tc0})]\label{fact:ampamp}
For any $m$-qubit circuit $C$ that  produces, 
\[C \ket{0^m} = \cos \theta \ket{\psi_0} \ket{0} + \sin \theta \ket{\psi_1} \ket{1},\]
for $k$ odd and $\theta = \frac{\pi}{2k}$,
the state $\ket{\psi_1} \ket{1}$ can be prepared from $\ket{0^m}$ using $O(k)$ applications of $C, C^\dag$ and single qubit unitaries.
\end{fact}
\noindent We will use a more convenient form of \cite{grier2026tc0}, which applies to any constant amplitude, stated below.
\begin{corollary}[Constant amplitude to exact]\label{cor:constamptoexact}
Let $\ket{\psi}$ be any $n$ qubit state. Suppose there is a depth $d$, $m$-ancilla $\qaczfkd$ circuit $C$ that, for some $\alpha \in [0,1]$, cleanly prepares the following $n+1$-qubit state $\ket{\varphi}$,
$$\ket{\varphi}_{T,t} = \sqrt{\alpha} \ket{\psi}_T \ket{1}_t + \sqrt{1-\alpha} \ket{\bad}_T \ket{0}_t,$$ 
Then, there exists a depth $d' = O(d/\alpha)$, $m' = m+2$-ancilla circuit $C'$ to cleanly prepare $\ket{\psi}$.
\end{corollary}
\begin{proof}
Let $\alpha = \sin \phi$ for some $\phi \in [0,\pi/2]$. Let $\theta \leq \phi$ be the largest value such that $\theta$ has the form  $\pi/(2k)$ for odd $k \in \mathbb{N}$. Then $\theta = \Theta(\phi)$. Let $\alpha' = \sin \theta$.
Using a fresh ancilla $a$, we can prepare the following state cleanly in depth $2d+1$, 
\begin{align}
    \ket{\varphi'} &= \ctrl{\R_{\alpha'/\alpha}}(t,a) \cdot \ket{\varphi}_{T,t} \ket{0}_a \\
&=  \sqrt{\alpha} \ket{\psi}_{T} \ket{1}_a \R_{\alpha'/\alpha} \ket{1}_t + \sqrt{1-\alpha} \ket{\bad}_T \ket{0}_a \ket{0}_t \\
    &= \sqrt{\alpha'} \ket{\psi}_{T} \ket{1}_a \ket{1}_t + \sqrt{1-\alpha'} \ket{\bad'}_{T,t} \ket{0}_a.
\end{align}
Then, since both $C, C^\dag$ are $\qaczfkd$ circuits, applying \cref{fact:ampamp} to $\ket{\varphi'} \ket{0^{m+2}}$ implies a depth $d' = O(d k)$, $m'$ ancilla $\qaczfkd$ circuit to prepare $\ket{\psi} \ket{11}_{t,a}$. The preparation can be made clean in the same depth using single qubit $X$ gates to flip $t,a$ registers. 
Note that $\sin \theta$ is concave on the interval $[0,\pi]$. This makes $\sin \theta = \Theta(\theta)$ and hence, $d' = O(d/\theta) = O(d/\phi) = O(d/\alpha)$. 
\end{proof}

\subsection{Useful distributions and associated quantum states}\label{sec:dist_prelim} 
For any distribution $\D$ on $m$-bit strings, we will use $\D^{\tens \ell}$ to refer to the distribution on $\ell \cdot m$ bit strings obtained by $\ell$ independent samples of $\D$. Additionally, for any $\D$, we use $\ham(\D)$ to refer to the distribution over $[m]$ induced by the Hamming weight of samples from $\D$. 

\begin{definition}[Bernoulli product distribution $\bern(\eps)^{\tens n}$ ]\label{def:bernoulli_dist}
The distribution of $n$ independent Bernoulli bits, each equal to $1$ with probability $\eps$.  
The associated state is $\ket{\eps}^{\tens n}$.
\end{definition}
\noindent Note that $\ham(\bern(\eps)^{\tens n})$ is just the binomial distribution $\binomd(n,\eps)$. 

\begin{definition}[Uniform $k$-hamming slice / Dicke $\D^n_k$]\label{def:dicke_dist}
The uniform distribution on all $n$-bit strings of hamming weight $k$. The associated state is $\ket{D^n_k}$, the $n$-qubit $k$-weight Dicke state,
    \[\ket{D^n_k} := \frac{1}{\sqrt{\binom{n}{k}}}\sum_{x \in \{0,1\}^n, |x| = k} \ket{x}.\]
\end{definition}
\noindent  The special case of Dicke states where $k=1$ is also called the $W$ state, and is denoted $\ket{W_n}$.
Finally, below is the reduction of \cite{gretta2026paritynotinqac0iff} that demonstrates tightness of our bounds with respect to the required $\FANOUT$. 
\begin{fact}[Dicke states imply $\FANOUT$ (Thm 4.14 of \cite{gretta2026paritynotinqac0iff})]
Suppose there exists a $\QAC(d,m)$ circuit $C$ that prepares the state $\ket{D^n_k}$, then, $\FANOUT_k \in \QACZ(O(d), m \cdot \poly(n))$.
\end{fact}
\noindent Note that $C$ is not necessarily clean nor exact, while our matching Dicke states upper bounds are both clean and exact, as is the ideal. 

\section{Limited-fanout state-synthesis toolkit}\label{sec:primitives}

In this section, we develop our constant-depth state-synthesis toolkit underlying our limited-fanout constructions. These are organized as,
\begin{enumerate}
\item $\QRAM$-based constructions for arbitrary $O(\log n)$ qubit states and arbitrary superpositions over one-hot vectors (\cref{sec:qramconst}). 
\item Amplitude manipulations tools beyond Grover (\cref{sec:ampmanip}).
\item Controlled-unitaries without general fanout (\cref{sec:ctrlunitaries}). 
\end{enumerate}

\subsection{Arbitrary quantum states in constant depth}\label{sec:qramconst} 
We will first show that arbitrary superpositions over $n$-qubit one-hot vectors can be constructed in plain $\QACZ$.
\begin{lemma}[Arbitrary one-hot superpositions in $\QACZ$]\label{lem:anyonehotqac}
Let $\D$ be any distribution over $[n]$ with pmf $p$. Then, the $n$-qubit state,
$$\ket{\wh{\D}} := \sum_{i\in[n]}\sqrt{p(i)}\ket{e_i},$$
where $e_i=0^{i-1}10^{n-i}$, can be constructed cleanly using $\poly(n)$ ancillae in $\QACZ$.
\end{lemma}
\begin{proof}
For each $i\in[n]$, let $q_i:= \frac{p(i)}{1+p(i)}$, and note $1-q_i = \frac{1}{1+p(i)}$. Recall that we define $\ket{\eps} := \sqrt{1-\eps} \ket{0} + \sqrt{\eps} \ket{1}$ for $\eps \in [0,1]$. 
First prepare the following product state using a layer of single qubit unitaries. 
\begin{align}
\ket{\psi_0}:=\bigotimes_{i\in[n]} \ket{q_i}
\end{align}
Now, the amplitude of each $\ket{e_i}$ in $\ket{\psi_0}$ is
\begin{align}
\braket{e_i | \psi_0} &=  \sqrt{q_i}\prod_{j\neq i}\sqrt{1-q_j} \\
&=\lr{\prod_{j\in[n]}\sqrt{1-q_j}}\sqrt{\frac{q_i}{1-q_i}}\\
&=\lr{\prod_{j\in[n]}\frac{1}{\sqrt{1+p(j)}}}\sqrt{p(i)}.
\end{align}
Letting,
\begin{align}
\alpha:=\prod_{j\in[n]}\frac{1}{1+p(j)},
\end{align}
observe that $\braket{e_i | \psi_0} = \sqrt{\alpha} \cdot \sqrt{p(i)}$. Therefore, the component of $\ket{\psi_0}$ inside the Hamming-weight-$1$ subspace is precisely $\sqrt{\alpha}\ket{\wh{\D}}$. 

Since $\sum_j p(j)=1$,
\begin{align}
2\leq 1/\alpha = \prod_{j\in[n]}(1+p(j))\leq \exp\lr{\sum_{j\in[n]}p(j)}=e,
\end{align}
and hence $\alpha = \Theta(1)$.  Now, using the $\EXACT_1$ gate, compute into a fresh ancilla $a$ whether the first register has Hamming weight exactly $1$, obtaining
\begin{align}
\ket{\psi_1}:=\sqrt{\alpha}\ket{\wh{\D}}\ket{1}_a+\sqrt{1-\alpha}\ket{\bad}\ket{0}_a.
\end{align}
Therefore, applying \cref{cor:constamptoexact} cleanly produces $\ket{\wh{\D}}$ in $O(1)$ depth using $\poly(n)$ ancillae in $\QACZ$.
\end{proof}

Then, using an indexing, this allows us to compress this one-hot vectors and obtain arbitrary $\log n$ qubit states. 

\begin{lemma}[One-hot encodings with QRAM]\label{cl:onehot}
 For any $n$, and $\ell = \lceil \log n \rceil$, the following reversible map on $i \in [n]$ can be implemented in $\QACZ[\QRAM_n]$,
 $$\ket{i} \ket{0^n} \mapsto \ket{0^\ell} \ket{e_i}$$
 where $e_i = 0^{i-1} 1 0^{n-i}$ for $i \in [k]$. 
\end{lemma}
\begin{proof}
First we will argue that the map $\ket{0^{\ell}}_A \ket{e_i}_Q \mapsto \ket{i}_A \ket{e_i}_Q$, is in $\QACZ$ and we will refer to the corresponding unitary as $U_0$, 
This follows by a known binary-indexing argument used in the exact $\EXACT_1$ construction of \cite{cleandicke2026}. For each $k\in[\ell]$, let
\begin{align}
S_k:=\{j\in[n]:\text{the $k$th bit of the binary representation of $j$ is $1$}\}.
\end{align}
Since $\ket{e_i}$ has Hamming weight at most $1$, the $k$th bit of $i$ is
\begin{align}
i_k=\bigvee_{j\in S_k}(e_i)_j.
\end{align}
Each qubit of $Q$ participates in at most $\ell= \lceil \log n\rceil$ such OR gates. Thus, using $\FANOUT_{\ell}$, we can compute the $\ell$ OR to produce each $\ket{i_k}$ in parallel.
We then reverse the fanout operations to clean the copied registers. Since $\FANOUT_\polylog$ can be done in $\QACZ$ (\cref{fact:kthresh}), this gives the claimed $\QACZ$ implementation.

Now, recall that $\QRAM_n$, acting on the index register $A$, memory register $Q = \clr{q_1, q_2 \dots q_n}$ and target $t$ produces, 
\begin{align}
\ket{i}_A\ket{x}_Q\ket{b}_t \mapsto \ket{i}_A\ket{x}_Q\ket{b\oplus x_i}_t.
\end{align}
Let $\wh{\QRAM}_n := (I_A \tens H^{\tens n+1}_{Q,t}) \cdot \QRAM_n \cdot (I_A \tens H^{\tens n+1}_{Q,t})$. Then, observe that, 
\begin{align}
 \QRAM_n \ket{i}_A\ket{x}_Q\ket{b}_t &= \sum_{j \in [n]} \delta_{ij} \cdot  \cnot(q_j, t) \ket{i}_A \ket{x}_Q \ket{b}_t  
\end{align}
and therefore, 
\begin{align}
    \wh{\QRAM}_n \ket{i}_A\ket{x}_Q\ket{b}_t &= \sum_{j \in [n]} \delta_{ij} \cdot  \lr{H^{\tens n+1}_{Q,t}\cnot(q_j, t) H^{\tens n+1}_{Q,t}} \ket{i}_A \ket{x}_Q \ket{b}_t  \\
    &= \sum_{j \in [n]} \delta_{ij} \cdot  \lr{H_{q_j}  \tens H_t \cdot \cnot(q_j, t) \cdot H_{q_j} \tens H_t} \ket{i}_A \ket{x}_Q \ket{b}_t  \\
    &= \sum_{j \in [n]} \delta_{ij} \cdot   \cnot(t, q_j)  \ket{i}_A \ket{x}_Q \ket{b}_t  \\
    &= \cnot(t,q_i) \ket{i}_A \ket{x}_Q \ket{b}_t
\end{align}
Therefore, $\wh{\QRAM}_n$ implements the map $\ket{i}_A\ket{x}_Q \ket{b}_t
\mapsto \ket{i}_A\ket{x\oplus b \cdot e_i}_Q\ket{b}_t$.

Putting it together, we now describe the claimed map. On input $\ket{i}_A \ket{0^n}_Q$ prepare a new ancilla $t$ in the $\ket{1}$ state to obtain, 
\begin{align}
    \ket{\psi_1(i)} &:= \wh{\QRAM}_n \ket{i}_A \ket{0^n}_Q \ket{1}_t \\
    &= \ket{i}_A \ket{e_i}_Q \ket{1}_t 
\end{align}
Then, apply the circuit $U_0$ in reverse to produce, 
\begin{align}
    \ket{\psi_2(i)} &:= (U_0)^\dag \ket{i}_A \ket{e_i}_Q \ket{1}_t \\
    &=  \ket{0^\ell}_A \ket{e_i}_Q \ket{1}_t
\end{align}
Finally flip the $t$ register with an $X$ gate to produce the claimed map cleanly.  
\end{proof}

\begin{corollary}\label{cor:anylogreal} 
 Let $\D$ be any distribution over $[n]$ with pmf $p$. Then, the $n$-qubit state, $\ket{\D} := \sum_{i \in [n]} \sqrt{p_i} \ket{i}$ can be prepared cleanly in $\QACZ[\QRAM_n]$ using $\poly(n)$ ancillae. 
\end{corollary}
\begin{proof}
First obtain the $n$-qubit state $\ket{\wh{D}}_Q$ from \cref{lem:anyonehotqac} with a plain $\QACZ$ circuit of $\poly(n)$ ancillae. Let $U$ be the $\QACZ[\QRAM_n]$ circuit to produce the map  \cref{cl:onehot}. Then, using a fresh ancilla register $\ket{0^{\ell}}_A$ apply $U$ in reverse to produce, 
\begin{align}\label{eq:here2654}
  U^\dag  \ket{0^{\ell}}_A \ket{\wh{D}}_Q &= \sum_{i \in [n]} \sqrt{p_i} \cdot U^\dag \ket{0^\ell}_A \ket{e_i}_Q \\
  &= \sum_{i \in [n]} \sqrt{p_i}] \ket{i}_A \ket{0^n}_Q \\
  &= \ket{D}_A \ket{0^n}_Q
\end{align}
as claimed.
\end{proof}

This gives us the following construction for general states. 

\begin{theorem}[Any $O(\log n)$-qubit state with QRAM]\label{thm:anylognqram}
Any $\ell = O(\log n)$-qubit state $\ket{\psi}$ can be cleanly prepared in $\QACZ[\QRAM_n]$ using $\poly(n)$ ancillae. Furthermore, any superposition $\ket{\psi_*}$ over $\clr{ \ket{e_i} }_{i \in [n]}$ can be prepared in $\QACZF$. 
\end{theorem}
\begin{proof}
It suffices to describe the preparation $\ket{\psi_*}$, for all $m$ using $\FANOUT_m$. Then, for $m= 2^{\ell} = \poly(n)$, the corresponding $\ell$-qubit state can be obtained by applying \cref{cl:onehot} exactly as in \cref{eq:here2654}.
Any such $\ket{\psi_*}$ can be described as follows for nonnegative and real $\alpha_i$, and $\phi_i \in [0, 2\pi)$. 
\begin{align}
    \ket{\psi_*} = \sum_{i \in [m]} e^{-i \phi_i} \sqrt{\alpha_i} \ket{e_i}
\end{align}
Let $\D$ be the distribution over $[n]$ given by probability masses $p(i) := \alpha_i$.  Using \cref{lem:anyonehotqac}, first prepare $\ket{\wh{D}}$ in $\QACZ$. 
Then, for each qubit, apply the single qubit unitary $U_i := \mathrm{diag}(1,e^{-i \phi_i})$ to correct the phases and obtain, 
\begin{align}
\lr{\bigotimes_{j} U_j} \ket{\wh{\D}} &=   \lr{\bigotimes_{j} U_j}  \sum_{i \in [m]} \sqrt{\alpha_i}  \ket{e_i} \\
    &= \sum_{i \in [m]} \sqrt{\alpha_i} U_i \ket{e_i} \\
    &= \sum_{i \in [m]} e^{-i \phi_i} \sqrt{\alpha_i} \ket{e_i}
\end{align}
which is exactly the state $\ket{\psi_*}$ as claimed.
\end{proof}

\subsection{Amplitude manipulations}\label{sec:ampmanip}
\noindent We first describe a useful primitive which enables the reasonable adjustment of amplitudes on $n$ marked branches of a state.  

\begin{lemma}[Adjust amplitudes of $n$ branches]\label{lem:adjustamp}
For any $n,k$, let $C$ be a depth-$d$ $m = \poly(n)$-ancilla $\qaczfkd$ circuit to produce the state $C \ket{0^m} = \ket{\psi}$ of the following form, 
\begin{align}
 \ket{\psi} = \sum_{i \in [n]} \sqrt{\alpha_i} \ket{e_i}_X \ket{\varphi_i}_T, 
\end{align}
where $X$ is an $n$-qubit register and  $\ket{\varphi_i}$ are arbitrary states on the (possibly empty) $T$ register. 
Let $\beta_1 \dots \beta_n$ be coefficients with each $0 \leq  \beta_i \leq 1$. Let $Z := \sum_{i} \alpha_i \beta_i$. 
Then, there exists a depth $d' = O(d/Z)$, $m' = \poly(n) \cdot m$ ancilla $\qaczfkd$ circuit $C'$ that cleanly prepares,
\begin{align}
\ket{\psi_*} := \frac{1}{\sqrt{Z}} \sum_{i \in [n]} \sqrt{\alpha_i \beta_i}  \ket{e_i}_X \ket{\varphi_i}_{T}.
\end{align}
In other words, $C' \ket{0^{m'}} = \ket{\psi_*} \ket{0^{m'-|T|-|X|}}$.
\end{lemma}
\begin{proof}
Label the qubits of $X$ as $x_1, x_2 \dots x_{n}$. 
Prepare $n$ new ancillae labeled $a_1 \dots a_{n}$ in the $\ket{1}$ state. The overall state is,
\begin{align}
    \ket{\psi_0} &= \sum_{i \in [n]} \sqrt{\alpha_i} \ket{0^{n-1}}_{X \setminus x_i} \ket{1}_{x_i} \ket{\varphi_i}_T \ket{1^{n}}_{a_1\dots a_{n}} 
\end{align}
For every $j \in [n]$, let $G_j := \ctrl{R_{(-\beta_j)}} (x_j,a_j)$. This is a valid $\QACZ$ gate that effects,
\begin{align}
    G_j \ket{e_i}_{X} \ket{1}_{a_j} &= \begin{cases}
       \ket{e_i}_{X} \ket{1}_{a_j} \qquad &\text{ if $i \neq  j$} \\ 
       \ket{e_i}_{X} \lr{\sqrt{\beta_j} \ket{1}_{a_j} + \sqrt{1-\beta_j} \ket{0}_{a_j}} & \text{ if $i = j$ }
    \end{cases} 
\end{align}
Then, in a single layer obtain, 
\begin{align}
        \ket{\psi_1} &:= \lr{\bigotimes_j G_j(x_j,a_j)} \cdot \ket{\psi_0} \\
        &= \sum_{i \in [n]} \sqrt{\alpha_i} \ket{e_i}_X \ket{\varphi_i}_T  \lr{ \sqrt{\beta_i} \ket{1^n}_{a_1 \dots a_{n}}  + \sqrt{1-\beta_i} \ket{\ov{e_i}}_{a_1 \dots a_{n}}}
\end{align}
where $\ov{e_i}$ denotes the $n$-bit string with all but the $i$th bit being $1$. Then, apply an $\AND$ gate to the $a_1 \dots a_n$ with a fresh ancilla $q$, to produce, 
\begin{align}
    \ket{\psi_2} &:= \sum_{i \in [n]} \sqrt{\alpha_i} \ket{e_i}_X \ket{\varphi_i}_T  \lr{\sqrt{\beta_i} \ket{1^n}_{a_1 \dots a_{n}} \ket{1}_q  + \sqrt{1-\beta_i} \ket{\ov{e_i}}_{a_1 \dots a_{n}} \ket{0}_q } \\
        &= \sqrt{Z} \cdot \ket{\psi_*} \ket{1^{n}} \ket{1}_q + \sqrt{1-Z} \ket{\bad} \ket{0}_q
\end{align}
where, recall that $Z = \sum_{i} \alpha_i \beta_i$. Now applying \cref{cor:constamptoexact}, produces in $O(d/Z)$ depth,
\begin{align}
    \ket{\psi_3} &:=  \ket{\psi_*} \ket{1^{n}}_{a_1 \dots a_n} \ket{0}_q 
\end{align}
and finally applying $X$ gates on each $a_j$ ancilla cleanly produces the desired state.  
\end{proof}

Next, we will need the following gadget derived from $\THR$ gates.

\begin{lemma}[Hamming weight computation gadget]\label{lem:hamgate}
For any $n,k \leq n$ and $x \in \bin^n$, let $h_k(x) := \min(k+1, |x|)$. Then, the following map $\HAM^n_k$ can be cleanly implemented by using $\poly(n)$ ancillae in $\qaczfkd$. 
   $$\ket{\x}_X \ket{0^{k+1}}_Y \mapsto \ket{\x}_X \ket{e_{h_k(x)}}_Y.$$
where we define $e_0 := 0^{k+1}$ and $e_j := 0^{j-1} 1 0^{k+1-j}$ for $j \in [k]$. 
\end{lemma}
\begin{proof}
First use $\FANOUT_k$ to obtain $k+1$ copies of the input register, labeled $X_1 \dots X_{k}, X_{k+1}$. Label the $k+1$ output qubits $Y$ as $y_1 \dots y_{k+1}$. 
For each $j \in [k]$ apply and $\EXACT^n_j(X_j, y_j)$ 
which is available to using threshold gates due to \cref{fact:kthresh} and for $j = {k+1}$, apply $\neg \THR_k^n(X_j,y_j)$ gate. 
This has the effect,
\begin{align}
 \ket{x}_{X_j} \ket{0} \mapsto \begin{cases} \ket{x}_{X_j} \ket{\blr{|x| = j}}_{y_j} & \text{(if $j \in [k]$)} \\
 \ket{x}_{X_j} \ket{\blr{|x| > k}}_{y_j} &\text{(if $j = k+1$)}\end{cases}.
\end{align}
The state on the $Y$ register is now precisely the objective $\ket{e_{h_k(x)}}$. Then, uncompute the registers $X_1 \dots X_{k+1}$ with another $\FANOUT$ application to cleanly produces the claimed map. 
\end{proof}

For completeness, we also include the following $\QACZF$ generalization of the parallel amplification technique of \cite{rosenthal2021qac0}, which appeared in an earlier version of this paper.

\begin{claim}[Parallel amplification with Fanout]
For any $n$-qubit state $\ket{\psi_*}$, suppose there exists a depth $d$, $m$ ancilla $\QACZF$ circuit $C$ to cleanly construct the state 
$$\ket{\varphi} := \sqrt{\alpha} \ket{\psi_*} \ket{1} + \sqrt{1-\alpha} \ket{0^n} \ket{0},$$
for some $\alpha \in (0,1]$ such that $\alpha^{-1} \leq \poly(n)$. Then there exists a $\QACZF$ circuit $C'$ of depth $d' = O(d)$ using $m' = O((m+n) /\alpha)$ ancillae that cleanly constructs $\ket{\psi_*}$.  
\end{claim}
\begin{proof}
Let $t = \lceil \alpha^{-1} \rceil$. First, starting with $(m + n) \cdot t$ ancillae, construct $t$ copies of $\ket{\varphi}$ in parallel, each on a $n$ qubit target register labeled $Q_i$ and a corresponding ancilla $a_i$. Let $A := \clr{a_1, a_2 \dots a_t}$ and $Q_{\all} := Q_1 \cup Q_2 \dots Q_t$. The resulting state can be written as, 
\begin{align}
    \ket{\psi_0}_{Q_{\all}, A} &:= \bigotimes_{i \in [t]} \lr{\sqrt{\alpha} \ket{\psi_*}_{Q_i} \ket{1}_{a_i} + \sqrt{1-\alpha} \ket{0^n}_{Q_i} \ket{0}} \\
    &= \sqrt{p_*} \lr{\frac{1}{\sqrt{t}} \cdot \sum_{i \in [t]} \ket{\psi_*}_{Q_i} \ket{1}_{a_i} \ket{0^{n\cdot (t-1)}}_{Q_{\all} \setminus Q_i} \ket{0^{t-1}}_{A \setminus a_i}} +  \sqrt{1-p_*} \ket{\bad} .
\end{align}
where $p_* = \Pr[\binomd(t,\alpha) = 1] = \Theta(\Pr[\binomd(t,1/t)]) = \Theta(1)$. 
Recall that we have the $\EXACT_1$ gate in $\QACZ$, apply this to $A$ with a fresh ancilla $a_0$, obtain, 
\begin{align}
\ket{\psi_1}_{Q_{\all}, A, a_0} := \sqrt{p_*} \lr{\frac{1}{\sqrt{t}} \sum_{i \in [t]} \ket{\psi_*}_{Q_i} \ket{1}_{a_i} \ket{0^{n(t-1)}}_{Q_{\all} \setminus Q_i} \ket{0^{t-1}}_{A \setminus a_i}} \ket{1}_{a_0} + \sqrt{1-p_*} \ket{\bad} \ket{0}_{a_0}.
\end{align}
Then apply \cref{cor:constamptoexact} to obtain in $O(d/p_*)$ depth, the state, 
\begin{align}
\ket{\psi_2}_{Q_{\all},A} :=  \frac{1}{\sqrt{t}} \sum_{i \in [t]} \ket{\psi_*}_{Q_i} \ket{1}_{a_i} \ket{0^{n \cdot (t-1)}}_{Q_{\all} \setminus Q_i} \ket{0^{t-1}}_{A \setminus a_i}.
\end{align}
Now prepare a fresh $n$-qubit target register $T$, and apply $t$ $\ctrl{\swap(a_i, Q_i, T)}$ gates, each swapping the registers $Q_i, T$ controlled on $a_i$ to obtain,  
\begin{align}
\ket{\psi_3}_{Q_{\all},A,T} &:= \lr{\prod_{i \in [t]}  \ctrl{\swap(a_i, Q_i, T)}} \cdot \ket{\psi_2} \ket{0^n}_T \\
&= \frac{1}{\sqrt{t}} \sum_{i \in [t]} \swap(Q_i,T) \ket{\psi_*}_{Q_i} \ket{1}_{a_i} \ket{0^{n \cdot (t-1)}}_{Q_{\all} \setminus Q_i} \ket{0^{t-1}}_{A \setminus a_i} \ket{0^n}_T\\
&= \ket{0^{nt}}_{Q_{\all}} \ket{W_t}_{A} \ket{\psi_*}_T.
\end{align}
Due to \cref{fact:ctrlswap} and $t = \poly(n)$, this transformation can be implemented in parallel in $O(1)$ depth using $\FANOUT_n$. 
Finally, since $\ket{W_t}$ can be cleanly computed in $\QACZ$ (\cref{fact:constdicke}), we run its circuit in reverse to uncompute the $A$ register and produce the state  $\ket{\psi_*}$ cleanly as claimed.
\end{proof}

\subsection{Controlled-unitaries}\label{sec:ctrlunitaries}
Recall that we cannot implement the controlled circuit operator, $\ctrl{C}$ in $\QACZ$ for arbitrary $\QACZ$ circuit $C$. However, this operator is available in $\QACZF$ for any $\QACZF$ circuit $C$, and hence we can scale down to obtain the following claim. 
\begin{claim}[Controlled-circuit in {$\QACZ[\FANOUT_\ell]$}]\label{cl:ctrlqac}
Let $C$ be a depth $d$, $m$ ancilla $\QACZ[\FANOUT_\ell]$ circuit that consists of at most $\poly(\ell)$ gates for some $\ell$. Then, the controlled version of $C$,
$\ctrl{C}$ can be implemented by a depth $O(d)$, $O(m + \poly(\ell))$ ancilla and $\poly(\ell)$ size $\QACZ[\FANOUT_\ell]$ circuit $C'$.
\end{claim}
\begin{proof}

Using nested $\FANOUT_\ell$ gates, copy the control to $\poly(\ell)$ ancillae. Then, each of these qubits can be used as controls for each of the $\poly(\ell)$ gates, without blowing up the depth (see \cref{sec:gatenotation}). Uncomputing the ancillae gives a clean implementation of $\ctrl{C}$. 
\end{proof}

\noindent This provides the following corollary for the construction of \cite{cleandicke2026}.
\begin{corollary}[Controlled-Dicke]\label{col:dickesuperpos}
For any $\ell$ there exists a $\qaczfk{\ell}$ circuit $C$ that cleanly implements the map for all $k \in [\ell]$,
$$\ket{e_k} \ket{0^{\ell}} \mapsto \ket{e_k} \ket{D^{\ell}_{k-1}},$$
where $e_k$ is the $\ell$-bit one-hot encoding of $k$.
\end{corollary}
\begin{proof}
Label the input qubits as $\ket{e_k}_{a_1 \dots a_{\ell}}$.
For each  $j \in \ell$, let $C^{\ell}_j$ be the circuit for constructing $\ket{D^{\ell}_{j-1}}$ in $\qaczfk{\ell}$ using $\poly(\ell)$ gates from \cref{fact:fanoutdicke}.
Then, from \cref{cl:ctrlqac}, we can implement $\ctrl{C^{\ell}_j}$ in $\qaczfk{\ell}$. Create $\ell$ registers $Q_1 \dots Q_{\ell}$ of $\ell$-qubits each and let $Q_{\all} = \bigcup_{j \in [\ell]} Q_j$. Apply a layer of $\ctrl{C^{\ell}_j}$ gates as below, 
\begin{align}
    \ket{\psi_0(k)} &:= \lr{\bigotimes_{j} \ctrl{C^{\ell}_j}(a_j, Q_j) } \cdot \ket{e_k}_{a_1 \dots a_{\ell}}  \ket{0^{\ell^2}}_{Q_{\all}} \\
   &=  \ket{e_k}_{a_1 \dots a_{\ell}} \ket{D^{\ell}_{k-1}}_{Q_k} \ket{0^{\ell (\ell-1)}}_{Q_{\all} \setminus Q_k}
\end{align}
Next, prepare another fresh $\ell$-qubit target register $T$, and apply  $\ctrl{\swap(Q_j,T)}$ gates for each $j$. This produces, 
\begin{align}
    \ket{\psi_1(k)} &:= \lr{\prod_{j \in [\ell]} I \tens \ctrl{\swap(Q_j,T)}} \ket{\psi_0(k)} \ket{0^{\ell}}_T \\
    &= \ket{e_k}_{a_1 \dots a_{\ell}} \ket{0^{\ell^2}}_{Q_{\all}} \ket{D^{\ell}_{k-1}}_T.
\end{align}
Although each gate has the same target $T$, the map only needs to work for inputs of hamming weight $1$. Thus this step can also be implemented in $\qaczfk{\ell}$ using \cref{fact:ctrlswap}, and $\ket{\psi_1(k)}$ is as claimed.
\end{proof}

\noindent We cannot simply apply \cref{cl:ctrlqac} to arbitrary $\QACZ$ circuits, even those with $\polylog(n)$ multi-qubit gates, such as the $W_n$ state constructions of \cite{grier2026tc0, cleandicke2026}, because they require $n$ single-qubit unitaries. We can, however implement a controlled version of the $W_n$ state and similar states using the following primitive. 

\begin{lemma}[Controlled state constructions]\label{lem:ctrlstate}
For any arbitrary $|T|$ qubit states, $\ket{\varphi_0}, \ket{\varphi_1}$, if the state $\ket{\psi}_{T,t} = \frac{1}{\sqrt{2}} \ket{\varphi_0}_T \ket{0}_t + \frac{1}{\sqrt{2}} \ket{\varphi_1}_T \ket{1}_t$ can be constructed by a depth $d$,  $\qaczfkd$ circuit $C_0$ then there exists a depth $3d+5$, $\qaczfkd$ circuit $C$, taking a single input bit $x$, that (cleanly) implements the following map:
$$\ket{0}_x\ket{0^m} \mapsto \ket{0}_x \ket{\varphi_0}, \quad \ket{1}_x \ket{0^m} \mapsto \ket{1}_x \ket{\varphi_1}.$$
\end{lemma}
\begin{proof}
Let $\ket{\alpha}$ be the following state, 
\begin{align}
   \ket{\alpha}_{T,a_0,a_1} = \frac{1}{\sqrt{2}} \ket{\varphi_0}_{T} \ket{0}_{a_0} \ket{1}_{a_1} - \frac{1}{\sqrt{2}} \ket{\varphi_1}_T \ket{1}_{a_0} \ket{0}_{a_1} 
\end{align}
This state can be prepared in depth $d+1$ circuit by applying a single $\cnot(a_0,a_1)$ gate to $\ket{\psi}_{T,a_0}\ket{1}_{a_1}$ constructed using $C_0$. Due to \cref{fact:reflstate} we can implement the reflection about $\ket{\alpha}$, $\R_{\alpha}(T, a_0, a_1)$, in depth $2d+3$. $\R_{\alpha}$ effectively swaps the $\ket{\varphi_0} \ket{01}$ and $\ket{\varphi_1} \ket{10}$ branches, i.e, for on any $b \in \bin$,
\begin{align}
    \R_{\alpha} \cdot \ket{\varphi_{b\oplus 1}} \ket{b \oplus 1} \ket{b} &=  (I - 2\kb{\alpha}) \cdot \ket{\varphi_{b \oplus 1}} \ket{b \oplus 1} \ket{b} \\
    &= \ket{\varphi_{b}} \ket{b} \ket{b \oplus 1}. \label{eq:swap99}
\end{align}
Now, to implement $C$, on input $b \in \bin$, first using $C_0$ on $|T|+1$ ancillae prepare, 
\begin{align}
    \ket{\psi_1(b)} &:= C_0 \ket{0^{|T|+1}} \ket{b}_x  \\
    &=  \frac{1}{\sqrt{2}} \ket{\varphi_0}_{T} \ket{0}_a \ket{b}_x  + \frac{1}{\sqrt{2}} \ket{\varphi_1}_{T} \ket{1}_a \ket{b}_x \\
    &= \frac{1}{\sqrt{2}} \ket{\varphi_{b\oplus 1}}_T \ket{b\oplus 1}_a \ket{b}_x + \frac{1}{\sqrt{2}} \ket{\varphi_b}_T \ket{b}_a \ket{b}_x.
\end{align}
The $\ket{b}_a\ket{b}_x$ branch is orthogonal to $\ket{\alpha}$, and thus unaffected by $\R_{\alpha}$. Next, apply $\R_{\alpha}(T,a,x)$, 
\begin{align}
    \ket{\psi_2(b)} &:= \R_{\alpha}(T,a,x) \ket{\psi_1(b)} \\
    &= \frac{1}{\sqrt{2}} \ket{\varphi_b}_T \ket{b}_a \ket{b \oplus 1}_x + \frac{1}{\sqrt{2}} \ket{\varphi_b}_T \ket{b}_a \ket{b}_x \tag{from \cref{eq:swap99}} \\
    &= \ket{\varphi_b}_T \ket{b}_a  \ket{+}_x.
\end{align}
Finally, the state on $b$ and $x$ can be cleaned up as,
\begin{align}
    \ket{\psi_3(b)} &:=  \swap(x,a) H_x \ket{\psi_2(b)} \\
    &= \swap(x,a) \ket{\varphi_b}_T \ket{b}_a  \ket{0}_x \\
    &= \ket{\varphi_b}_T \ket{0}_a \ket{b}_x.
\end{align}
$\swap(x,a) = \cnot(x,a) \cnot(a,x)$, thus the described $C$ has depth $3d+5$.
\end{proof}

\begin{corollary}\label{cor:constctrl}
For any $n$-qubit state $\ket{\perp_n}$ satisfying $\braket{\perp_n | 0^n} = 0$, and $\alpha \in (0,1)$ satisfying both $\alpha, 1-\alpha = \Theta(1)$, if the state $\ket{\psi} = \sqrt{\alpha} \ket{0^n} + \sqrt{1-\alpha} \ket{\perp_n}$ can be prepared in $\qaczfkd$, then the following map is in $\qaczfkd$:

$$\ket{0}_x \ket{0^n} \mapsto \ket{0}_x \ket{0^n}, \quad \ket{1}_x \ket{0^n} \mapsto \ket{1}_x \ket{\perp_n}.$$
\end{corollary}
\begin{proof}
Let $\ket{\varphi} = \frac{1}{\sqrt{2}} \ket{0^n} + \frac{1}{\sqrt{2}} \ket{\perp_n}$. Due to \cref{lem:ctrlstate}, it suffices to prepare $\ket{\varphi}$. Let $\gamma = \alpha/(1-\alpha)$.
We will describe the case when $\gamma \leq 1$ (i.e $\alpha \leq 1/2$), the other case is symmetric and follows by replacing the $\OR$ gates with $\NOR$ and $\R_{\gamma}$ with $\R_{1/\gamma}$ below.  

First, prepare the state $\ket{\psi_0} := \ket{\psi}_T \ket{01}_{a',a}$ with a new ancillae $a,a'$ and compute the $\OR$ of the $T$ register onto $a'$ to produce, 
$\ket{\psi_1} := \sqrt{\alpha} \ket{0^n}_T \ket{0}_{a'} \ket{0}_a + \sqrt{1-\alpha} \ket{\perp_n}_T \ket{1}_{a'} \ket{0}$. 
Then, apply a $\ctrl{\R_{\gamma}}(a',a)$ gate and uncompute the $\OR$, obtain, $\ket{\psi_2} := \sqrt{2\alpha} \ket{\varphi} \ket{01}_{a',a} + \sqrt{1-2\alpha} \ket{\perp_n} \ket{00}_{a',a}$.
This marks the equal superposition of $\ket{0^n}, \ket{\perp_n}$ and amplifying using \cref{cor:constamptoexact} produces $\ket{\varphi}$ cleanly in $O(1/\alpha) = O(1)$ depth. 
\end{proof}

\section{Constant-depth symmetric states} \label{sec:dicke}

In this section, we prove the main theorem of our paper, that one can synthesize $\ket{D^n_k}$ in $\QACZ[\FANOUT_k]$. Throughout this section, assume that $k$ is at least a sufficiently large constant. For ease of exposition, we do not optimize the constants in our analysis.

\subsection{Intermediate distributional states}
\noindent Recall that for a distribution $\D$ over $\bin^n$, its associated state $\ket{\D}$ is an $n$-qubit state whose amplitudes in the standard basis are given by,
\begin{align}
   \braket{\x|\D} = \sqrt{\Pr_{\y \sim \D} [\y = \x] }.
\end{align}
We will frequently use the notion from \cref{sec:dist_prelim} for the rest of this section. 
Any symmetric distribution $\D$ is given by a convex combination of Dicke distributions thus can be uniquely defined by the probability masses of $\ham(\D)$. Analogously, the corresponding state $\ket{\D}$, can be described as a superposition of Dicke states.  
Next we will define an intermediate distribution that forms the key piece of our main construction. 

\begin{definition}[Damped Binomial Distribution ($\S^m_k$)]\label{def:damped_binomial}
For $m,k \in \mathbb{N}$, $k \leq m$, let $\S^m_k$ be the symmetric distribution over $\bin^m$ defined by probability masses $s(j)$ of $\ham(\S^m_k)$ supported on $j \in [k]$, as, 
$$s(j) = \lam \cdot \frac{1}{(mk)^j} \cdot \binom{m}{j}$$
where $\lam$ is chosen to normalize $s(j)$ correctly. 
The associated state is, 
\[\ket{\S^m_k} := \sum_{j \in [k]} \sqrt{s(j)} \ket{D^m_j}.\]
\end{definition}

\noindent First we will show that $\S^m_k$ has the following useful properties.
\begin{claim}[Bounds on the normalization constant]\label{cl:hybridnorm}
For any $m, 1 \leq k\leq m$, let $\lam$ be the normalization constant from \cref{def:damped_binomial}. 
$$\frac{k}{k+1} \cdot k \le \lam \le k.$$
\end{claim}
\begin{proof}
By definition,
\begin{align}
\lam^{-1} &= \sum_{r=1}^k \binom{m}{r}\left(\frac{1}{mk}\right)^r  \geq \binom{m}{1}\frac{1}{mk} = \frac{1}{k}.
\end{align}
Hence, $\lam \le k$.

For the upper bound, use $\binom{m}{r} \le \frac{m^r}{r!}$ to get
\begin{align}
\lam^{-1} &\le \sum_{r=1}^k \frac{1}{r!k^r} \le \sum_{r=1}^\infty \frac{1}{r!k^r} = e^{1/k} - 1 \leq 1/k + 1/k^2.
\end{align}
Rearranging yields $\lam \ge k \cdot \frac{k}{k+1}$.
\end{proof}

\begin{claim}[pmf of $\S^m_k$ is dominated by a binomial]\label{cl:binomdom}
For any $m \geq 2$, $1 \leq k \leq m/2$ , let $s(j)$ be as defined in \cref{def:damped_binomial}, then, there exists a universal constant $c$ such that  
$$s(j) \leq c \cdot \Pr[\binomd(m,1/m) = j].$$
\end{claim}
\begin{proof}
For $j \geq 1$,

\[ s(j) = \lambda \cdot (mk)^{-j} \cdot \binom{m}{j} \leq \binom{m}{j}m^{-j}  \leq 4 \cdot \Pr[\binomd(m,1/m) = j].\]

\end{proof}

\begin{claim}[Hamming slice of $\S^m_k$ is uniform over support ]\label{cl:shiftedoccuniform}
For any $k \geq 1$ and $1 \leq j \leq k$  let $\D$ be the distribution over $\bin^{mj}$  defined by $\D := (\S^m_k)^{\tens j}$. Then, for any $x \in \mathrm{supp}(\D)$,  $\Pr[\D = x]$ depends only on $|x|$, the Hamming weight of $x$. 
\end{claim}
\begin{proof}
Each string $x \in \mathrm{supp}(\D)$ can be written as $x = x_1 \dots x_j$ for $x_i \in \bin^m$ with $h_i := |x_i| \in [k]$. Since each $\S^m_k$ is symmetric, 
\begin{align}
   \Pr[\D = (x_1, x_2 \dots x_j)] &= \prod_{i \in [j]} \Pr[\S^m_k = x_i]  \\
 &= \prod_{i \in [j]} \binom{m}{h_i}^{-1} \cdot s(h_i) \label{eq:here13243} 
\end{align}
Substituting from \cref{def:damped_binomial} into \cref{eq:here13243} for $\tau = 1/(mk)$. 
\begin{align}
   \Pr[\D = (x_1, x_2 \dots x_j)] &=  \lam^j \prod_{i \in [j]} \binom{m}{h_i}^{-1} \cdot \binom{m}{h_i} \cdot \tau^{h_i} \\
   &= \lam^{j} \prod_{i \in [j]} \tau^{h_i} \\ 
   &= \lam^j \cdot \tau^{|x|} 
\end{align}
which only depends on $|x|$ and is as claimed.
\end{proof}

We require a controlled version of the associated state to leverage \cref{cor:constctrl}. To do so, we will show that a superposition over $\ket{0^m}$ and $\ket{\S^m_k}$ can be prepared.


\begin{lemma}\label{lem:zeroshft}
For any $1 \leq k \leq m$, there exists $\gamma \in (0,1)$ with $\gamma/(1-\gamma) = \Theta(1)$, such that the state $\ket{\psi} := \sqrt{1-\gamma} \ket{0^m} + \sqrt{\gamma} \ket{\S^m_k}$ can be prepared in $\qaczfkd$. 
\end{lemma}
\begin{proof}
First prepare $\ket{\psi_0} = \ket{1/m}^{\tens m}$. Then, using the $\HAM^m_k$ map from  \cref{lem:hamgate} on a fresh $k+1$-qubit register $Y$, obtain, 
\begin{align}
    \ket{\psi_1} &= \sum_{i = 0}^k \sqrt{a_i} \ket{D^m_i}_T \ket{e_{i}}_Y + \sum_{i = k+1}^m \sqrt{a_i} \ket{D^m_i} \ket{e_{k+1}}_Y.
\end{align}
where $a_i = \Pr[\binomd(m,1/m) = i]$ and recall $e_0 := 0^{k+1}$.
Since $\Pr[\binomd(m,1/m) \leq k] = \Omega(1)$, we can apply an $X$ gate on the last qubit of the $Y$ register followed by \cref{fact:ampamp} to obtain in $O(1)$ depth, 
\begin{align}\label{eq:statewewant}
 \ket{\psi_2} &= \sum_{i = 0}^k \sqrt{\alpha_i} \ket{D^m_i}_T \ket{e_{i}}_Y,
\end{align}
where $\alpha_i = a_i / \Pr[\binomd(m,1/m) \leq k]$.

Let $s(j)$ be the probability mass function of $\ham(\S^m_k)$ as defined in \cref{def:damped_binomial}. For each $j \in [k]$, let $\beta_j = s(j)/(4 \cdot \alpha_j)$ and $\beta_0 = 1$. Then, $\alpha_i \geq a_i$ and due to \cref{cl:binomdom}, each $\beta_j \in [0,1]$ for $j \in [k]$. Additionally, 
$\sum_{j=0}^k \alpha_j \beta_j = 4^{-1} + \alpha_0 = \Theta(1)$. Therefore, applying \cref{lem:adjustamp} produces the following state $\ket{\psi_3}$ in depth $O(1)$.  
\begin{align}
\ket{\psi_3} := \sqrt{\gam} \cdot \sum_{j \in [k]} \sqrt{s(j)} \ket{D^m_j} \ket{e_{j+1}}_Y + \sqrt{1-\gam} \ket{0^m} \ket{e_{1}}_Y,
\end{align}
where $\frac{\gamma}{1-\gamma} = \sum_{j \in [k]} \alpha_j \beta_j / \alpha_0 = \frac{1}{4 \cdot \alpha_0} + 1 = \Theta(1)$, as required.
Finally, we can apply $\HAM^m_{k}$ again to obtain, 
\begin{align}
    \ket{\psi_4} &:= \sqrt{\gam} \cdot \sum_{j \in [k]} \sqrt{s(j)} \ket{D^m_j} \ket{0^{k+1}}_Y + \sqrt{1-\gam} \ket{0^m} \ket{0^{k+1}}_Y  \\
    &= \lr{ \sqrt{\gam} \ket{\S^m_k} + \sqrt{1-\gamma} \ket{0^{m}} } \tens \ket{0^{k+1}}_Y,
\end{align}
as claimed.
\end{proof}

\begin{corollary}[Controlled dampened binomial state]\label{cor:ctrlhybrid}
For any $k$, $m \geq k$, the following map on an input qubit $x \in \bin$ can be implemented in $\QACZ$,   
    $$\ket{0}_x \ket{0^m} \mapsto \ket{0}_x \ket{0^m}, \quad \ket{1}_x \ket{0^m} \mapsto \ket{1}_x \ket{\S^m_k}$$
\end{corollary}
where $\ket{\S^m_k}$ is the state associated with the shifted-truncated binomial.
\begin{proof}
The required map is obtained through \cref{cor:constctrl} due to \cref{lem:zeroshft}. 
\end{proof}

\subsubsection{Dicke occupancy states}\label{sec:occdistr}
\noindent We now describe the preparation of the state corresponding to the Dicke occupancy distribution, induced on each bucket of the Dicke state.
First, we define the occupancy of a bistring distribution under bucketing.

\begin{definition}[Binomial occupancy distribution, $\Occ(\D,\ell)$]\label{def:occupancy_dist}
  Given a distribution $\D$ on $\bin^n$, and $\ell$ s.t $\ell \mid n$, 
  $\Occ(\D,\ell)$ is the distribution over $\{0,\ldots,\ell\}$ as described below. Partition $n$ bits into $\ell$ buckets of sizes $n/\ell$ each. A bucket is \emph{empty} if all its bits are $0$.  For $\x \in \bin^n$, define $\occ(\x)$ to be the number of non-empty buckets in $\x$. Then,  
$$\Pr[\Occ(\D,\ell) = k] := \Pr_{\x \sim \D} \blr{\occ(\x)=k}.$$
\end{definition}
\noindent When $\D$ is the Dicke distribution, define the corresponding state to track the occupancy as below. 
\begin{definition}[Dicke occupancy state, $\ket{\Occ(\D^n_k,\ell)}$]\label{def:occstate}
Let $\D^n_k$ denote the uniform distribution over $n$-bit, Hamming weight $k$ strings. For $\ell \mid n$ and $k\leq \ell$, this is an $(n + k)$-qubit state corresponding to the occupancy of $\dkn$ with $\ell$ buckets given by,  
\begin{align*}
    \ket{\Occ(\D^n_k,\ell)} &= \binom{n}{k}^{-1/2} \sum_{\x \in \{0,1\}^n, |\x| = k} \ket{\x} \ket{e_\occ(\x)}  \\
                          &= \sum_{j = 0}^k \sqrt{p(j)} \dmkj  \ket{e_j}_A,
\end{align*}
where $m = n/\ell$ and $\dmkj$ is the state conditioned on measuring occupancy $j$, and $e_s$ denotes the $k$-bit string with all but the $j$th bit zero, and $p$ is the pmf of $\Occ(\mathcal{D}^n_k, \ell)$. 
\end{definition}
\noindent Note that $\dmkj$ is symmetric within each bucket, and also symmetric across buckets whenever $n \equiv 0 \pmod{\ell}$, and symmetric across the first $\ell-1$ buckets otherwise.

Our strategy will be to approximate $\occknl$ by using the map from \cref{cor:ctrlhybrid} and then correct the amplitudes. First, will relate the two distributions as follows.

\begin{claim}[Hybrid to Dicke occupancy ratio]\label{cl:hybridtoocc} 
For any $m$, $1 \leq k \leq m$, and $\ell \geq k^3$, let $p(j)$ be the probability mass function of $\Occ(\D^{m \cdot \ell}_k, \ell)$. For $j \in [k]$, let
$$q(j) := \Pr_{\x \sim (\S^m_k)^{\tens j}} \bigl[|\x| = k\bigr].$$
It holds that,
\[q(k)=\Theta(1),p(k) = \Theta(1) \quad \text{and} \quad \frac{p(j)}{q(j)} \leq e^2 \cdot k^{j-k} \]
\end{claim}
\begin{proof}
Let $v = \frac{1}{mk}$. Define, 
\begin{align}
    \Gamma_j := \sum_{\substack{t_1, t_2 \dots t_j \in [k]\\t_1 + t_2 \dots + t_j = k}} \prod_{i \in [j]} \binom{m}{t_i}
\end{align}
Then, we can write $p(j)$ and $q(j)$ in terms of this quantity as, 
\begin{align}
    q(j) &=  \sum_{\substack{t_1, t_2 \dots t_j \in [k]\\t_1 + t_2 \dots t_j = k}} \prod_{i \in [j]} \left(\binom{m}{t_i}  \cdot \lam \cdot v^{t_i} \right) \\
    &= \lam^j \cdot v^k \cdot \Gamma_j, 
\end{align}
and 
\begin{align}
   p(j) &= \binom{\ell}{j} \cdot \binom{m\ell}{k}^{-1} \cdot \Gamma_j.
\end{align}
Therefore, letting $n = m \cdot \ell$,
\begin{align}
    \frac{p(j)}{q(j)} &= \binom{\ell}{j} \cdot \binom{n}{k}^{-1} \cdot \lam^{-j} \cdot v^{-k}  \\
     &= \binom{\ell}{j} \cdot \binom{n}{k}^{-1} \cdot \lam^{-j} \cdot (mk)^{k} \\
     &\leq \binom{\ell}{j} \cdot \binom{n}{k}^{-1} \cdot k^{-j} \cdot (mk)^{k} \cdot \left( \frac{k + 1}{k} \right)^{j} \tag{from \cref{cl:hybridnorm}} \\
     &\leq e \cdot \frac{\ell^j}{j!} \cdot \binom{n}{k}^{-1} \cdot k^{k-j} \cdot m^{k}.  \label{eq:here234}
\end{align}
Note,
\begin{align}
    \binom{n}{k} \geq \frac{(n-k)^k}{k!} = \frac{n^k}{k!} \cdot \lr{1-\frac{k}{n}}^{k}.
\end{align}
Substituting into \cref{eq:here234}, 
\begin{align}
\frac{p(j)}{q(j)} &\leq e\cdot \frac{\ell^j}{j!} \cdot  \frac{k!}{n^k} \cdot \lr{1-\frac{k}{n}}^{-k} \cdot k^{k-j} \cdot m^{k} \\
&= e\cdot \frac{k!}{j!} \cdot \ell^{j-k} \cdot k^{k-j} \cdot \lr{1-\frac{k}{n}}^{-k} \\
&\leq e\cdot \ell^{-k+j} \cdot k^{2(k-j)} \cdot \lr{1-\frac{k}{n}}^{-k} \\
&\leq e\cdot k^{-3(k-j)} \cdot k^{2(k-j)} \cdot \lr{1-\frac{k}{n}}^{-k} \tag{$\ell^{-1} \leq k^{-3}$} \\
&\leq e\cdot k^{j-k} \cdot \lr{1-\frac{1}{mk^2}}^{-k} \\
&\leq e^2 \cdot k^{j-k}
\end{align}
where in the last step, we used that $k < m$. To see $q(j) = \Theta(1)$, note that its upper bounded by $1$, and lower bounded by $(1 - 1/(k+1))^k$. To see $p(j) = Theta(1)$, note that this occupancy distribution stochastically dominates its ``balls-and-bins'' version, and so standard arguments give $p(j) = \Omega(1)$.
\end{proof}

\noindent We now describe the construction.
\begin{lemma}[Binomial occupancy states in $\QACZ$]\label{lem:binomoccckt}
Given $k$, $\ell \geq k^3$ such that $\ell \mid n$, the state $\occknl$ can be constructed cleanly in $\qaczfkd$ using $\poly(n)$ ancillas.
\end{lemma}
\begin{proof}
Recall, $\occknl := \sum_{j = 0}^k \sqrt{p(j)} \ket{e_j}_A \ket{D^{m,\ell}_{k,j}}_T$. Label the $m$ qubits in each bucket as $T_j$ so that $T = T_1 \cup T_2 \dots T_\ell$, and label the qubits in the $A$ register as  $A = \clr{a_1, a_2 \dots a_k}$. 

Let $q(j)$ be as defined in \cref{cl:hybridtoocc}, $r(j) := p(j) / q(j)$ and  $R := \sum_{j \in [k]} r(j)$. Perform the following steps in $\qaczfkd$.
\begin{enumerate}
\item Set $\delta_j = \sqrt{r(j)/R}$ and use \cref{lem:anyonehotqac} to prepare on the $A$ register, 
\begin{align}
    \ket{\psi_0}_{A} := \sum_{j \in [k]} \delta_j \ket{e_j}_A 
\end{align}
\item  Let $C_1$ be the circuit from \cref{col:dickesuperpos}. With a new register, $B = \clr{b_1,b_2\dots b_\ell}$, obtain,  
\begin{align}
    \ket{\psi_1} &:= C_1 \ket{\psi_0} \ket{0^{\ell}}_{B}  \\
    &= \sum_{j \in [k]} \delta_j \cdot \ket{e_j}_{A} \dicke{j}{\ell}_B 
\end{align}
\item  Let $m = n/\ell$ and recall the state $\S^m_k$ from \cref{def:damped_binomial}. 
Let $C_2$ be the map from \cref{cor:ctrlhybrid}, then, for each bucket $i \in [\ell]$ prepare $\ket{\S^m_k}$ on $T_i$ controlled on $b_i$ as below.  
\begin{align}
    \ket{\psi_2}_{A,B,T} &:= \lr{ \bigotimes_{i \in [\ell]} \ctrl{C_2}(b_i,T_i)} \cdot \ket{\psi_1} \ket{0^n}_T \\ 
    &= \sum_{j \in [k]} \delta_j \cdot \ket{e_j}_{A} \sum_{\y \in \bin^\ell : |\y| = j}  \ket{\y}_{B} \lr{ \bigotimes_{i : y_i = 1} \ket{\S^m_k}_{T_i} } \lr{ \bigotimes_{i : y_i = 0} \ket{0^{m}}_{T_i} }. 
\end{align}
The nonzero-bits in each branch of $\ket{D^\ell_j}$ will indicate which $j$ of the $\ell$ buckets are non-empty. This essentially prepares the state $\ket{\S^m_k}$ on each non-empty bucket indicated by the bits of $\y$. 

\item Note $\braket{\S^m_k|0^m} = 0$. Then, using an OR gate on $T_i$ registers with $b_i$ as the target, uncompute $b_i$ to obtain, 
\begin{align}
    \ket{\psi_3} &= \sum_{j = 1}^k \delta_j \ket{e_j}_{A} \ket{0^\ell}_B \sum_{\substack{S \subseteq [\ell] \\ |S| = j}} \lr{ \bigotimes_{i \in S} \ket{\S^m_k}_{T_i} }
\lr{ \bigotimes_{i \in [\ell]\setminus S} \ket{0^{m}}_{T_i} } \label{eq:symm0}  
\end{align}
 Observe that due to \cref{cl:shiftedoccuniform}, all strings $x \in \bin^{m\ell}$ with hamming weight $k'$ and $\occ(x) = j$ appear with the same amplitude. Therefore, the state can be re-written as, 
\begin{align}
\ket{\psi_3} &=  \sum_{j = 1}^k \delta_j \ket{e_j}_A  \ket{0^\ell}_B \sum_{k' = j}^k \gamma_{j,k'} \ket{\D^{m,\ell}_{j,k'}}_{T} \label{eq:symm1}
\end{align}
where $\ket{\D^{m,\ell}_{j,k'}}$ is the state conditioned on measuring occupancy $j$ conditioned on hamming weight $k'$ as defined in \cref{def:occstate} $\gamma^2_{j,k'}$ is its corresponding probability. From here on we will drop the uncomputed $B$ register.
\item Recall $r(j), R$ as defined previously with $\delta_j = \sqrt{r(j)/R}$, and observe $\gamma_{j,k} = \sqrt{q(j)}$. Apply $\THRESHOLD^n_k$ gates, via \cref{fact:kthresh} on the $T$ register and a new ancilla $a$ to obtain, 
\begin{align}\label{eq:beforeamp}
    \ket{\psi_4}_{A,T,a} &= \lr{\sum_{j = 1}^k \delta_j \gamma_{j,k} \ket{e_j} \ket{\D^{m,\ell}_{j,k}}} \ket{1}_a + \gamma_{\bad} \ket{\bad} \ket{0}_a  \\
    &= \frac{1}{\sqrt{R}} \lr{\sum_{j = 1}^k \sqrt{p(j)} \ket{e_j} \ket{\D^{m,\ell}_{j,k}}} \ket{1}_a + \gamma_{\bad} \ket{\bad} \ket{0}_a \\
    &= \sqrt{1/R} \occknl + \sqrt{1-1/R} \cdot \ket{\bad} \ket{0}_a.
\end{align}
\end{enumerate}
Let $d_0 = O(1)$ be the depth of the circuit so far to obtain $\ket{\psi_4}$. 
Due to \cref{cl:hybridtoocc}, $R = \Theta(1)$. Thus, applying \cref{cor:constamptoexact} produces $\occknl$ exactly as claimed. 
\end{proof}

\subsection{Dicke States}

Now we have the required pieces to prove the main theorem. 
\kdicke
\begin{proof}
Let $\ell = k^3$, we will describe the construction when $\ell \mid n$.  
Partition the $n$ target qubits $T$ into $\ell$ buckets, where $T_i$ is the set of $m = n/\ell$ qubits in the $i$th bucket ($T = \bigcup_{i \in [\ell]} T_i$).
Then, from \cref{lem:binomoccckt}, we can construct the associated state with $\Occ(\D^n_{k}, \ell)$, 
\begin{align}
\occknl &= \sum_{j = 1}^k \sqrt{p(j)} \dmkj_T \ket{e_j}_{a_1 \dots a_k}
\end{align}
We can prepare this state due to \cref{lem:binomoccckt}. To obtain $\dkn$ cleanly from this, use temporary ancillae to mark whether each bucket is non-empty using an $\OR$ gate, then use $\HAM^{\ell}_k$ to uncompute the $A$ register, and uncompute the marked buckets using another layer of $\OR$ gates. Discarding the uncomputed register produces 
\begin{align}
    \ket{\psi_2} :=  \sum_{j \in [k]} \sqrt{p(j)} \dmkj,  
\end{align}
which, by \cref{def:occstate} is exactly the state $\dkn$. 

\paragraph{General case:} When $\ell \nmid n$, let $m = \lceil n/\ell \rceil$ and $n' = \ell \cdot m$. First construct $\dicke{k}{n'}_T$ as above. Now use an additional ancilla $a$ to compute the $\NOR$ on the last $n'-n$ qubits, obtaining, 
\begin{align}
    \ket{\psi_3} := \sqrt{p_0} \dkn_T \ket{1}_a + \sqrt{1-p_0} \ket{\bad}_T \ket{0},
\end{align}
where $p_0$ is the probability of measuring $0$ on the last $n'-n < m$ qubits.  Observe that since $m \ll n/k$, this happens with overwhelming probability. Hence, we can apply \cref{cor:constamptoexact} to obtain $\dkn_T$ in $O(1)$ depth. 
\end{proof}

\noindent Then, from \cref{fact:kthresh}, since $\qaczfk{\polylog(n)} = \QACZ$ we immediately have the following.

\polylogdicke*

\subsection{General symmetric states}

Finally we extend to any superposition of these Dicke states as below. 
\ksymmstate
The proof of this theorem follows by a first constructing the corresponding superposition over the $\ket{e_i}$ using \cref{thm:anylognqram}.  Then, using a very similar procedure as the one in \cref{lem:binomoccckt}, we produce the corresponding superposition over $\dmkj$ as defined in \cref{def:occstate}. Finally, we obtain the corresponding Dicke state by uncomputing the occupancies and hamming weights similar to  \cref{thm:kdicke}. 

Even when provided black-box access to a controlled preparation of $\ket{D^n_k}$, preparing arbitrary superpositions over $\ket{D^n_j}$ for $j \leq k$ naively requires  
$\FANOUT_n$. We instead construct the corresponding superposition over occupancy states from scratch using the same ideas as before.  We provide the full proof in \cref{sec:arbitrarysymm}. 

\ifnum\ANON=1
\else
\section{Acknowledgments}
The authors are grateful to Avishay Tal and Ryan Williams for helpful discussions and advice. 
\fi

\appendix
\bibliographystyle{alpha}
\bibliography{main, dicke_priors}

\section{Deferred Proofs}\label[appendix]{sec:arbitrarysymm}
In this section we will formally prove \cref{thm:any_symm}.
First, we need the following lemmas. 

\begin{lemma}[General Hybrid to Dicke occupancy ratio]\label{cl:genhybridtoocc} 
For any $m$, $1 \leq k_* \leq m$, $\ell = k_*^3$, $k \in [k_*]$, and $j \in [k]$, let $p_k(j)$ be the probability mass function of $\Occ(\D^{m \cdot \ell}_k, \ell)$ and let
$$q_k(j) := \Pr_{\x \sim (\S^m_{k_*})^{\tens j}} \bigl[|\x| = k\bigr].$$
Then, for $R_k := \sum_{j \in [k]} \frac{p_k(j)}{q_k(j)}$ it holds that $R_k \leq c$ for a global constant $c$.   
\end{lemma}

\begin{proof}
Let $v = \frac{1}{mk}$. First observe that $\S^m_k$ is obtained by conditioning $\S^m_{k_*}$ on Hamming weight $\leq k$. The probability that a single sample from $\S^m_{k_*}$ has Hamming weight $\leq k$ is at least $\Pr[\ham(\S^m_{k_*}) = 1] \geq \frac{k_*}{k_* + 1}$. 
Let $\lam_{j,k}$ be the probability that $j$ samples from $\S^m_{k_*}$ all have Hamming weight $\leq k$. Then, 
\begin{align}
    \lam_{j,k} &\geq \lr{ \frac{k_*}{k_* + 1}}^j\\
    &\geq e^{-2j/(k_* + 1)}\\
    &\geq e^{-2j/k}.
\end{align}

Therefore, we can re-write $q_k$ as, 
\begin{align}
    q_k(j) &= \lam_{j,k} \cdot \Pr_{\x \sim (\S^m_{k})^{\tens j}} \bigl[|\x| = k\bigr] \\
    &\geq e^{-2} \cdot \Pr_{\x \sim (\S^m_{k})^{\tens j}} \bigl[|\x| = k\bigr].
\end{align}
Therefore, using \cref{cl:hybridtoocc}, 
\begin{align}
    \frac{p_k(j)}{q_k(j)} &\leq e^2 \frac{p_k(j)}{\Pr_{\x \sim (\S^m_{k})^{\tens j}} \bigl[|\x| = k\bigr] } \\
    &\leq e^4 k^{j-k}.
\end{align}
summing over all $k$, this gives, 
\begin{align}
    R_k &\leq e^{4} \sum_{d = 0}^{k-1} \frac{1}{k^d} \\
    &\leq 2e^4.
\end{align}
\end{proof}

\begin{lemma}\label{lem:customthresh}
For $n, k \leq n$, the following map defined for $x \in \bin^n, 0 \leq j \leq k$, can be implemented in $\qaczfkd$. 
$$\ket{e_j} \ket{x} \ket{0} \mapsto \ket{e_j} \ket{x} \ket{\blr{|x| \leq j}}$$
\end{lemma}
\begin{proof}
On input $\ket{e_j}_A \ket{x}_X$, make $k$ copies of the $X$ register and label them $X_1 \dots X_k$.  Let $A = \clr{a_1 \dots a_k}$. 
Prepare a fresh $k$ qubit ancilla register $y_1, y_2 \dots y_k$. Using the $i$th copy $X_i$, compute $\THRESHOLD^n_i$ on $y_i$ to obtain,
\begin{align}
    \ket{\psi_1(j,x)} := \ket{e_j} \ket{x}_{X} \ket{x}^{\tens k}_{X_1 \dots X_k} \bigotimes_{i \in [k]} \ket{[|x| \leq i]}_{y_i} 
\end{align}
Next, on an output register $q$, compute the $\OR$ of the $\AND(y_i,a_i)$. This is a standard procedure that one can implement in cleanly in depth $3$ using $k$ additional ancillae. This produces, 
\begin{align}
    \ket{\psi_2(j,x)} := \ket{e_j} \ket{x}_{X} \ket{x}^{\tens k}_{X_1 \dots X_k} \lr{\bigotimes_{i \in [k]} \ket{[|x| \leq i]}_{y_i} } \ket{|x| \leq j}_q
\end{align}
Finally, uncompute the $y_i$ registers and then the $X_1 \dots X_k$ to produce, 
\begin{align}
    \ket{\psi_3(j,x)} := \ket{e_j} \ket{x}_{X} \ket{0^{kn}}_{X_1 \dots X_k} \ket{0^k}_{y_1 \dots y_k} \ket{|x| \leq j}_q
\end{align}
as claimed. 
\end{proof}

\subsection{Limited-fanout symmetric states synthesis}
We wish to prepare in $\qaczfkd$ the state,
    $$\ket{\psi_*} := \sum_{k = 0}^{k_*} \eta_k \ket{D^n_k}$$
for arbitrary amplitudes $\eta_1, \eta_2 \dots \eta_{k_*} \in \mathbb{C}$ with $\sum_k |\eta_k|^2 = 1$. 
\begin{proof}[Proof of \cref{thm:any_symm}]
Let $\ell := (k_*)^3$. First we will describe the case  when $\ell \mid n$. For ease of exposition define $\I := \clr{(j,k)}_{k \in [k_*], j \in [k]} \cup \clr{(0,0)}$ and define $\S$ as shorthand for the distribution  $\S^m_{k_*}$ from \cref{def:damped_binomial}. 

\paragraph{Notation.}
For all $k \in [k_*]$, let  $p_k(j)$ be the probability mass function of $\occlabel$.
For a $k_*$ qubit register $A = {a_1, a_2 \dots a_{k_*}}$ and $n$ qubit register $T$ define the following  state for $k \in [k_*]$ 
\begin{align}
    \ket{\O^n_k}_{A,T} &:= \sum_{j = 0}^k \sqrt{p_k(j)} \ket{e_j}_A \dmkj_T. 
\end{align}
Observe that this state is simply $\occknl \tens \ket{0^{k_*-k}}$. (For $k = 0$ define $\ket{\O^n_k} := \ket{0^{n+k_*}}$). Then, given amplitudes $\eta_1, \eta_2 \dots \eta_{k_*}$, it suffices to construct the following state
\begin{align}\label{eq:objoccstateap}
    \ket{\varphi_*} := \sum_{k = 0}^{k_*} \eta_k \ket{e_k}_Q \ket{\O^n_k}_{A,T}, 
\end{align}
where $Q$ is a $k_*$-qubit register. 

For $k \in [k_*]$ and $j \in [k]$, define $p_k(j), q_k(j)$ and $R_k$ as in \cref{cl:genhybridtoocc}. To handle edge cases define $e_0 := 0^{k_*}$ and $R_0 = 0$. From \cref{cl:genhybridtoocc}, observe that, 
\begin{align}
    R &:= \sum_{k = 0}^{k_*} R_k \cdot |\eta_k|^2 \\
    &\leq \sum_{k = 0}^{k_*} |\eta_k|^2 \cdot O(1) \\
    &= O(1)\label{eq:normalize12323ap}
\end{align}

\paragraph{Procedure for constructing $\ket{\varphi_*}$.}  For all $k \in [k_*], j \in [k]$ set $\delta^k_j = \eta_k \cdot \sqrt{r_k(j)/R}$ and set $\delta^0_0 = \eta_0/\sqrt{R}$.  
Begin with the following state whose preparation we will describe at the end 
\begin{align}
    \ket{\psi_0} := \sum_{(k,j) \in \I} \delta^k_j \ket{e_k}_Q \ket{e_j}_A. 
\end{align}

\begin{enumerate}
\item  Let $C_1$ be the circuit from \cref{col:dickesuperpos}. With a new register, $B = \clr{b_1,b_2\dots b_\ell}$, obtain,  
\begin{align}
    \ket{\psi_1} &:= C_1 \ket{\psi_0} \ket{0^{\ell}}_{B}  \\
    &= \sum_{(j,k) \in \I} \delta^k_j \cdot \ket{e_k}_Q \ket{e_j}_{A} \dicke{j}{\ell}_B. 
\end{align}
\item  Let $m = n/\ell$ and let $C_2$ be the map from \cref{cor:ctrlhybrid}, then, for each $i$ prepare $\ket{\S}$ on $T_i$ controlled on $b_i$ as below.
\begin{align}
    \ket{\psi_2}_{A,B,T} &:= \lr{ \bigotimes_{i \in [\ell]} \ctrl{C_2}(b_i,T_i)} \cdot \ket{\psi_1} \ket{0^n}_T \\ 
    &= \sum_{(j,k) \in \I} \delta^k_j \cdot \ket{e_k}_Q \ket{e_j}_{A} \sum_{\y \in \bin^\ell : |\y| = j}  \ket{\y}_{B} \lr{ \bigotimes_{i : y_i = 1} \ket{\S}_{T_i} } \lr{ \bigotimes_{i : y_i = 0} \ket{0^{m}}_{T_i} }. 
\end{align}

\item Note $\braket{\S|0^m} = 0$. Then, using an OR gate on $T_i$ registers with $b_i$ as the target, uncompute $b_i$ to obtain, 
\begin{align}
    \ket{\psi_3} &= \sum_{(j,k) \in \I} \delta^k_j \ket{e_k}_Q \ket{e_j}_{A} \ket{0^\ell}_B \sum_{\substack{S \subseteq [\ell] \\ |S| = j}} \lr{ \bigotimes_{i \in S} \ket{\S}_{T_i} }
\lr{ \bigotimes_{i \in [\ell]\setminus S} \ket{0^{m}}_{T_i} }. \label{eq:symm12ap}  
\end{align}

\item Due to \cref{cl:shiftedoccuniform}, all strings $x \in \bin^{m}$ with Hamming weight $k'$ and $\occ(x) = j$ appear with the same amplitude. Therefore the state can be re-written as $\ket{\psi} = \ket{\psi_3'} \ket{0^{\ell}}_B$, 
\begin{align}
\ket{\psi_3'} &=  \sum_{(j,k) \in \I } \delta^k_j \ket{e_k} \ket{e_j}_A   \sum_{k' = j}^k \gamma_{j,k'} \ket{\D^{m,\ell}_{j,k'}}_{T},
\end{align}
where $\ket{\D^{m,\ell}_{j,k'}}$ is the state conditioned on measuring occupancy $j$ and hamming weight $k'$ as defined in \cref{def:occstate} and $\gamma^2_{j,k'}$ is its corresponding probability. We will drop the $B$ register in the following. 

\item Recall that $q_k(j) := \Pr[|\S^{\tens j}| = k]$, and is precisely the amplitude on $\ket{e_k}_Q \ket{e_j}_A \ket{D_{j,k}^{m,\ell}}_T$. Marking this on ancilla $a$ for all $j,k$ using \cref{lem:customthresh} produces, 
\begin{align}
\ket{\psi_4} &= \lr{\sum_{(j,k) \in \I} \delta^k_j \sqrt{q_k(j)} \ket{e_k}_Q  \ket{e_j}_A \ket{\D^{m,\ell}_{j,k}}_T \ket{1}_a } + \gamma_{\bad} \ket{\bad} \ket{0}_a \\
&=  \lr{\frac{1}{\sqrt{R}} \sum_{(j,k) \in \I} \eta_k \sqrt{p_k(j)} \ket{e_k}_Q  \ket{e_j}_A \ket{\D^{m,\ell}_{j,k}}_T \ket{1}_a } + \gamma_{\bad} \ket{\bad} \ket{0}_a \\
&=  \lr{\frac{1}{\sqrt{R}} \sum_{k = 0}^{k_*} \eta_k \sqrt{p_k(j)} \ket{e_k}_Q  \ket{\O^n_k} \ket{1}_a } + \gamma_{\bad} \ket{\bad} \ket{0}_a 
\end{align}
\end{enumerate}
Let $d_0 = O(1)$ be the depth of the circuit so far to obtain $\ket{\psi_4}$. 
Due to \cref{cl:hybridtoocc}, $R = O(1)$. Thus, applying \cref{cor:constamptoexact} produces the desired state $\ket{\varphi_*}$ from \cref{eq:objoccstateap} in $O(d_0 R) = O(1)$ depth.

\paragraph{Construction of the initial state.}
It remains to describe the construction of $\ket{\psi_0}$. 
Let $b = \lceil \log k_* \rceil$. Then, for $j,k \in [2^{b}]$ define $v(k,j) := k \cdot 2^{b} + j$. Observe that, $\ket{v(k,j)} = \ket{k} \ket{j}$, (recall that for an integer $v$, $\ket{v}$ denotes its binary representation in the standard basis). 
Therefore, due to \cref{cor:anylogreal}, we can prepare the $2b$-qubit superposition, 
\begin{align}
    \sum_{(j,k) \in \I} \delta^k_j \cdot \ket{k} \ket{j}
\end{align}
and then obtain $\ket{\psi_0}$ by applying \cref{cl:onehot}. 

\paragraph{To corresponding Dicke superposition.}
Once we have $\ket{\varphi_*}$, use \cref{lem:hamgate} to uncompute the $Q$ register and uncomputing the $A$ register the same way as in \cref{thm:kdicke}, to produce, 
\begin{align}
    \sum_{(j,k) \in \I} \eta_k \sqrt{p_k} \ket{0^k}_Q \ket{0^k}_A \dmkj_T &= \sum_{k} \eta_k \ket{D^n_k}_T \ket{0^{2k}}_{Q,A} \\
    &= \ket{\psi_*} \ket{0^{2k}}_{Q,A}
\end{align}
as required. 

\paragraph{Generalizing to arbitrary $n$}
When $\ell \nmid n$, let $m = \lceil n/\ell \rceil$ and $n' = \ell \cdot m$. 
Let $p_{k,0}$ be the probability of measuring $0$ on the last $n'-n < m$ qubits of $\dkn$. Regardless of $k$, we have, $p_{k,0} \geq 1/2$. Let $\eta'_k = ( \eta_k \cdot p_{k,0})/\sqrt{Z}$ where $Z = \sum_{k} \eta_k p_{k,0} \leq 2$, and construct, 
\begin{align}
    \ket{\psi_1}_T := \sum_{k \in [k_*]} \eta'_k \ket{D^{n'}_k}_T
\end{align}
using the above. 
Now use an additionally ancilla $a$ to compute the $\NOR$ on the last $n'-n$ qubits labelled $T'$, to obtain 
\begin{align}
    \ket{\psi_2} := \sqrt{1/Z} \ket{\psi_*}_{T\setminus T'} \ket{0^{n'-n}}_{T'} \ket{1}_a  
    + \sqrt{1-1/Z} \ket{\bad}_T \ket{0},
\end{align}
Then, using \cref{fact:ampamp}, we obtain $\ket{\psi_*}$ in $O(Z) = O(1)$ depth.
\end{proof}


\end{document}